\documentclass[sigconf, screen, dvipsnames, authorversion]{acmart}

\usepackage{multirow}
\usepackage{amsmath}
\usepackage{amsthm}
\usepackage[bb=dsserif,bbscaled=1.2]{mathalpha}
\usepackage{bm}
\usepackage{mathtools}
\usepackage[inline]{enumitem} 
\usepackage{subcaption}
\usepackage{xcolor}

\usepackage{wrapfig}

\usepackage{arydshln}
\makeatletter
\def\adl@drawiv#1#2#3{
        \hskip.5\tabcolsep
        \xleaders#3{#2.5\@tempdimb #1{1}#2.5\@tempdimb}%
                #2\z@ plus1fil minus1fil\relax
        \hskip.5\tabcolsep}
\newcommand{\cdashlinelr}[1]{%
  \noalign{\vskip\aboverulesep
          \global\let\@dashdrawstore\adl@draw
          \global\let\adl@draw\adl@drawiv}
  \cdashline{#1}
  \noalign{\global\let\adl@draw\@dashdrawstore
          \vskip\belowrulesep}}
\makeatother

\newtheorem{assumption}{Assumption}

\usepackage{colortbl}
\newcolumntype{p}{>{\columncolor{gray!10}}r}
\newcolumntype{R}{S[table-format=1.2]}

\usepackage{siunitx}

\usepackage{tikz}
\usetikzlibrary{automata, arrows, bayesnet, bending}

\DeclareMathOperator{\argmax}{arg\,max}

\DeclareMathOperator{\argsort}{arg\,sort}

\copyrightyear{2024}
\acmYear{2024}
\setcopyright{acmlicensed}\acmConference[KDD '24]{Proceedings of the 30th ACM SIGKDD Conference on Knowledge Discovery and Data Mining}{August 25--29, 2024}{Barcelona, Spain}
\acmBooktitle{Proceedings of the 30th ACM SIGKDD Conference on Knowledge Discovery and Data Mining (KDD '24), August 25--29, 2024, Barcelona, Spain}
\acmDOI{10.1145/3637528.3671687}
\acmISBN{979-8-4007-0490-1/24/08}

\begin{document}
\title{On (Normalised) Discounted Cumulative Gain as an Off-Policy Evaluation Metric for Top-$n$ Recommendation}

\author{Olivier Jeunen}
\affiliation{
  \institution{ShareChat}
  \city{Edinburgh}
  \country{United Kingdom}
}
\author{Ivan Potapov}
\affiliation{
  \institution{ShareChat}
  \city{London}
  \country{United Kingdom}
}
\author{Aleksei Ustimenko}
\affiliation{
  \institution{ShareChat}
  \city{London}
  \country{United Kingdom}
}

\begin{abstract}
Approaches to recommendation are typically evaluated in one of two ways:
\begin{enumerate*}
    \item via a (simulated) online experiment, often seen as the \emph{gold standard}, or
    \item via some offline evaluation procedure, where the goal is to approximate the outcome of an online experiment.
\end{enumerate*}
Several offline evaluation metrics have been adopted in the literature, inspired by ranking metrics prevalent in the field of Information Retrieval.
(Normalised) Discounted Cumulative Gain (nDCG) is one such metric that has seen widespread adoption in empirical studies, and higher (n)DCG values have been used to present new methods as  the \emph{state-of-the-art} in top-$n$ recommendation for many years.\looseness=-1

Our work takes a critical look at this approach, and investigates \emph{when} we can expect such metrics to approximate the gold standard outcome of an online experiment.
We formally present the assumptions that are necessary to consider DCG an \emph{unbiased} estimator of online reward and provide a derivation for this metric from first principles, highlighting where we deviate from its traditional uses in IR.
Importantly, we show that normalising the metric renders it \emph{inconsistent}, in that even when DCG is unbiased, ranking competing methods by their normalised DCG can invert their relative order.
Through a correlation analysis between off- and on-line experiments conducted on a large-scale recommendation platform, we show that our \emph{unbiased} DCG estimates strongly correlate with online reward, even when some of the metric's inherent assumptions are violated.
This statement no longer holds for its normalised variant, suggesting that nDCG's practical utility may be limited.
\end{abstract}



\maketitle

\section{Introduction \& Motivation}
Recommender systems power algorithmic decision-making on platforms across the web.
They occur in many different application domains, but all are centred around a similar question:
\begin{center}
    ``\emph{What} content do we recommend to \emph{whom}?''
\end{center}

Research and applications in the field of recommender systems have undergone a large shift in the last few decades, moving from \emph{rating}~\cite{Bennet2007} to \emph{item} prediction~\cite{Steck2013} and, more recently, embracing an \emph{interventionist} view~\cite{Joachims_London_Su_Swaminathan_Wang_2021}.
Throughout this evolution, the evaluation practices that are commonly adopted in the field have undergone a parallel shift.
Indeed, whereas the continuous and explicit nature of ratings lends itself to the Root Mean Square Error (RMSE) metric; discrete item predictions are more often viewed in a \emph{ranking} setup where Information Retrieval (IR) metrics like Precision, Recall, and (normalised) Discounted Cumulative Gain (nDCG) have been widely adopted~\cite{Valcarce2020}; and the interventionist view lends itself particularly well to \emph{counterfactual} and \emph{off-policy} estimation techniques that can be directly mapped to online metrics~\cite{Vasile2020, Saito2021}.

Despite the fact that evaluation methods are a core topic in recommender systems research that enjoy a significant amount of interest, some problems remain that are fundamentally hard to solve.
Online experiments (i.e. randomised controlled trials or A/B-tests) are seen as the \emph{gold standard} of evaluation practices, and deservedly so: by leveraging interactions with users, they allow us to directly measure an array of online metrics for a given recommendation model~\cite{Kohavi2020}.
Nevertheless, as they are costly to conduct and the academic research community seldom has access to platforms with real users~\cite{Jeunen2019DS}, \emph{offline} evaluation practices are a common alternative used to showcase newly proposed methods' performance, both in the research literature and in industry applications (often as a precursor to an online experiment).

Even though the need for offline evaluation methods that mimic the outcome of an online experiment is clear, the reality is that existing methods seldom do so satisfactorily~\cite{Beel2013,Garcin2014,Rossetti2016,Jeunen2018}, even if advances in counterfactual estimation techniques have recently led to several success stories~\cite{Gilotte2018,Gruson2019}.
The reasons for this ``offline-online mismatch'' are manifold --- and some can be attributed due to offline evaluation inconsistencies (even before we compare them to online results).
First, there is a fundamental mismatch between many ``next-item prediction'' metrics (e.g. recall) and online metrics (e.g. click-through rate).
Although this can be partially alleviated by the interventionist lens~\cite{Jeunen2019REVEAL_EVAL}, a majority of published research remains focused on IR-inspired metrics.
Second, a myriad of evaluation options can lead to contrasting results~\cite{Canamares2020}, and several offline metrics differ in robustness and discriminative power~\cite{Valcarce2020}.
Third, sampled versions of these metrics have been adopted for efficiency reasons, but are \emph{inconsistent} with their unsampled counterparts and should thus be avoided~\cite{Krichene2020,Li2020,Canamares2020RecSys}.
Fourth, multiple recent works have reported troubling trends in the reproducibility of widely cited methods~\cite{FerrariDacrema2019,FerrariDacrema2021,Rendle2020,Rendle2022,Rendle2019,Cavenaghi2023}---similar issues plagued the adjacent IR field a decade earlier~\cite{Armstrong2009}. 

This article aims to contribute to this important line of research, by focusing on the widely adopted (normalised) Discounted Cumulative Gain metric~\cite{Jarvelin2002}.
Focusing on the purpose that offline metrics serve, we ask a simple question:

\begin{center}    
``\emph{When can we expect (n)DCG to accurately approximate\\the gold standard outcome of an online experiment?}''
\end{center}

The main scientific contributions we present in pursuit of answering this question, are the following:
\begin{enumerate}
    \item We formally present the assumptions that are necessary to consider DCG an \emph{unbiased} estimator of online reward, providing a derivation for this metric from first principles whilst linking it to off-policy estimation (\S~\ref{sec:dcg_unbiased}).
    \item We formally prove that the widespread practice of \emph{normalising} the DCG metric renders it inconsistent with respect to DCG, in that the ordering given by nDCG can differ from that given by DCG, and provide empirical evidence (\S~\ref{sec:ndcg_inconsistent}).
    \item Empirical results from off- and online experiments on a large-scale recommendation platform show that:
    \begin{enumerate}
        \item the unbiased DCG metric strongly correlates with online metrics over time, whereas nDCG does not (\S ~\ref{sec:RQ1}),
        \item whilst differences in online metrics directionally align with differences in both nDCG and DCG, the latter can enjoy improved sensitivity to detect statistically significant online improvements (\S ~\ref{sec:RQ3}).
    \end{enumerate}
    \item We revisit the assumptions that are necessary to consider DCG an unbiased estimator: discussing when they can be reasonable and how we can relax them, giving rise to a research agenda and ways forward (\S ~\ref{sec:beyond}).
\end{enumerate}
\section{Background \& Related Work}\label{sec:relwork}
Offline evaluation methods for recommender systems have been studied for decades~\cite{Herlocker2004}, and their shortcomings are widely reported.
\citeauthor{Canamares2020} provide an overview of common approaches, highlighting how different choices (in pre-processing, metrics, splits,...) lead to contrasting results~\cite{Canamares2020}.
More problematic, \citeauthor{Ji2023} show that common train-test-split procedures lead to data leakage issues that affect conclusions drawn from offline experiments~\cite{Ji2023}.
Other recent work shows that sampled versions of evaluation metrics that only rank a sample of the item catalogue instead of the full catalogue, are inconsistent with the full metrics, leading the authors to explicitly discourage their use~\cite{Krichene2020,Li2020,Canamares2020RecSys}.
Even when we manage to steer clear from these pitfalls, biases in logged data can give rise to undesirable phenomena like Simpson's paradox~\cite{Jadidinejad2021,Jeunen2023_CONSEQUENCES}.
General de-biasing procedures have been proposed to this end~\cite{LYang2018}, as well as \emph{off-policy} estimation techniques~\cite{su2019cab, Su2020_ICML, Ma2020, Kiyohara2022} and methods to evaluate competing estimators~\cite{Saito2021_Robustness,Su2020_Adaptive,Udagawa2022}.

Most traditional ranking evaluation metrics stem from IR.
\citeauthor{Valcarce2020} find that nDCG offers the best discriminative power among them~\cite{Valcarce2020}.
Findings like this reinforce the community's trust in nDCG, and it is commonly used to compare novel top-$n$ recommendation methods to the state-of-the-art, also in reproducibility studies~\cite{FerrariDacrema2019,FerrariDacrema2021,Rendle2020,Rendle2022}.
\citeauthor{Ferrante2021} argue that while nDCG can be preferable because it is bounded and normalised, problems can arise because the metric is not easily transformed to an interval scale~\cite{Ferrante2021}.
They all do not consider the consistency of (n)DCG.

Other recent work highlights that commonly used online evaluation metrics relying on \emph{experimental} data (e.g. click-through rate), differ fundamentally from commonly used offline evaluation metrics that rely on \emph{organic} interactions (e.g. hit-rate)~\cite{Jeunen2019DS,Jeunen2019REVEAL_EVAL}.
\citeauthor{Deffayet2023} argue that a similar mismatch is especially pervasive when considering reinforcement learning methods~\cite{Deffayet2023}, and \citeauthor{Diaz2021} emphasises critical issues with interpreting organic implicit feedback as a user preference signal~\cite{Diaz2021}.
These works together indicate a rift between online and offline experiments.

Several open-source simulation environments have been proposed as a way to bypass the need for an online ground truth result~\cite{rohde2018recogym,Ie2019Recsim,Saito2021OBP}, and several works have leveraged these simulators to empirically validate algorithmic advances in bandit learning for recommendation~\cite{JeunenKDD2020,Jeunen2021A,Jeunen2021B,Sakhi2020,Bendada2020}.
Nevertheless, whether conclusions drawn from simulation results accurately reflect those drawn from real-world experiments, is still an open research question.

In this work, we focus on the \emph{core} purpose that offline evaluation metrics serve: to give rise to offline evaluation methodologies that accurately mimic the outcome of an online experiment.
To this end, we focus on the widely used (n)DCG metric, and aim to take a step towards closing the gap between the off- and online paradigms.
\section{Formalising the problem setting}\label{sec:problemsetting}
Throughout, we represent the domain for a random variable $X$ as $\mathcal{X}$ and a specific instantiation as $x$, unless explicitly mentioned otherwise.
We deal with a session-based feed recommendation setup, describing a user's journey on the platform as a trajectory $\tau$.
Contextual features describing a trajectory are encoded in $x \in \mathcal{X}$, which includes features describing the user $u \in \mathcal{U}$ and possible historical interactions they have had with items on the platform.
In line with common notation in the decision-making literature, we will refer to these items as \emph{actions} $a \in \mathcal{A}$.
As is common in real-world systems, the size of the item catalogue (i.e. the action space $|\mathcal{A}|$) can easily grow to be in the order of  hundreds of millions, prohibiting us to score and rank the entire catalogue directly.
This is typically dealt with through a two-stage ranking setup, where a more lightweight \emph{candidate generator} stage is followed by a \emph{ranking} stage that decides the final personalised order in which we present items to the user~\cite{VanDang2013, Covington2016, Ma2020}.

We adopt generalised probabilistic notation for two-stage rankers in this work, but stress that our insights are model-agnostic and directly applicable to single-stage rankers as well. 

Let $\mathcal{A}^{k}$ denote all subsets of $k$ actions: $\mathcal{A}^{k} \subseteq 2^{\mathcal{A}}: \forall a^{k} \in \mathcal{A}^{k}, |a^{k}|=k$.
A candidate generation policy $\mathcal{G}$ defines a conditional probability distribution over such sets of \emph{candidate} actions, given a context:
$
    \mathcal{G}(A^{k}|X)\coloneqq\mathsf{P}(A^{k}|X,\mathcal{G}).
$
We will use the shorthand notation $\mathcal{G}(X)$ when context allows it.
Note that this general notation subsumes other common scenarios, such as candidate generation policies that are deterministic, consist of ensembles, or simply yield the entire item catalogue $\mathcal{A}$ (i.e. \emph{single-stage} ranking systems).

After obtaining a set of candidate items for context $x$ by sampling $a^{k} \sim \mathcal{G}(x)$, we pass them on to the ranking stage.

In line with our probabilistic framework, we have a ranking policy $\mathcal{R}$ that defines a conditional probability distribution over rankings (i.e. permutations of $a^{k}$).
Define $\bm{\sigma}$ as a possible permutation over $A^{k}$.
Formally, we have:
$
    \mathcal{R}(\bm{\sigma}|A^{k}, X)\coloneqq\mathsf{P}(\bm{\sigma}|A^{k}, X,\mathcal{R}).
$
When context allows it, we will use shorthand notation $\mathcal{R}(X)$ to absorb the candidate generation step. Then, the ranker is given by: 
\begin{equation}
    \mathcal{R}(X) \coloneqq \sum_{a^{k} \in \mathcal{A}^{k}}\mathsf{P}(\bm{\sigma}|A^{k}=a^{k}, X,\mathcal{R})\mathsf{P}(A^{k}=a^{k}|X,\mathcal{G}).
\end{equation}

Now, for a given context $x$, we can obtain rankings over actions by sampling $\sigma=(a_{1},\ldots,a_{k}) \sim \mathcal{R}(x)$.
These rankings are then presented to users, who scroll through the ordered list of items and \emph{view} them along the way.
We will not \emph{yet} restrict our setup to a specific user model that describes \emph{how} users interact with rankings, but introduce a binary random variable $V$ to indicate whether a user has viewed a given item.\footnote{Note that this problem setting deviates from traditional work dealing with web search in IR, where multiple items are shown on screen and item-specific view events cannot be disentangled trivially.
In contrast, our items take up most of the user's mobile screen when presented, and we are able to deduce accurate item-level view labels from logged scrolling behaviour. See e.g.~\citet{Jeunen2023_C3PO} for work dealing with this setting.}
As such, we will assume that logged trajectories only contain items that were viewed by the user (i.e. $V=1$).
That is, if the user abandons the feed after action $a_{i}$ shown at rank $R=i$, we do not log samples for actions $a_{j}, \forall j>i$.

Users can not only view items, but they can interact with them in several ways.
These interaction signals could be seen as the \emph{reward} or \emph{(relevance) label}.
Following traditional notation where rewards are \emph{clicks} we will denote them by random variable $C$.
Nevertheless, these signals are general and can be binary (e.g. likes), real-valued (e.g. revenue), or higher-dimensional to support multiple objectives (e.g. diversity, satisfaction, and fairness~\cite{Mehrotra2018,Mehrotra2020}).

As users can interact with every item separately, we define $C^{k} = (c_1, \ldots, c_{k})$ as the logged rewards over all ranks.
We do not place any restrictions on the reward distribution yet, so the reward at any given rank can be dependent on actions at other ranks: $(c_{1}, \ldots, c_{k})\sim\mathsf{P}(C^{k}|X,A_{1}, \ldots, A_{k})$.
Now, a trajectory consists of contextual information, as well as a sequence of user-item interactions with observed reward labels: $\tau = \{x_{\tau}, (a_1,c_{1}),\ldots(a_{t},c_{t})\}$, where $|\tau|=t$.
The true metric of interest that we care about is the expectation of reward, over contexts sampled from an unknown marginal distribution $x\sim\mathsf{P}(X)$, candidates sampled from our candidate generator $a^{k} \sim \mathcal{G}(x)$, rankings sampled from our ranker $(a_{1},\ldots,a_{k})\sim\mathcal{R}(a^{k},x)$, and rewards sampled from the unknown reward distribution $(c_{1}, \ldots, c_{k})\sim\mathsf{P}(C^{k}|X=x,A_{1}=a_{1}, \ldots, A_{k}=a_{k})$.
We will denote with $C$ the sum of rewards over all observed ranks: $C = \sum_{i=1}^{|\tau|}c_{i}$.
Note that this notation generalises typical evaluation metrics that are used in real-world recommender systems, such as per-item dwell-time or counters of engagement signals.

In order to obtain an estimate for our true metric of interest $\mathbb{E}[C]$, we can perform an online experiment.
Indeed, in doing so, we effectively sample from the above-mentioned distributions and obtain an empirical estimate of our metric by averaging observed samples.
For a dataset $\mathcal{D}_{0}$ containing logged interactions under policies $\mathcal{G}_{0}$ and $\mathcal{R}_{0}$, Eq.~\ref{eq:online_exp} shows how to obtain this empirical estimate:
\begin{equation}\label{eq:online_exp}
    \mathop{\mathbb{E}}\limits_{(a_{1},\ldots,a_{k}) \sim \mathcal{R}_{0}(x)}[C] \approx \frac{1}{|\mathcal{D}_{0}|\cdot|\tau|}\sum_{\tau \in \mathcal{D}_{0}}\sum_{i=1}^{|\tau|}c_{i}.
\end{equation}
As we directly measure the quantity of interest that depends on the deployed policies $\mathcal{G}_{0}$ and $\mathcal{R}_{0}$, this online estimator is often seen as the gold standard.
Nevertheless, it is costly to obtain.
Indeed, as has been widely reported~\cite{Jeunen2019DS,Gilotte2018,Larsen2023}, online experiments require us to:
\begin{enumerate*}
    \item bring hypotheses for new policies up to production standards for an initial test,
    \item wait several days or weeks to deduce statistical significant improvements, and
    \item possibly harm user experience when the policies are performing subpar.
\end{enumerate*}
Furthermore, several pitfalls arise when setting up or interpreting results from online A/B experiments~\cite{Kohavi2022,Jeunen2023_misassumption}.

For these reasons, we want to be able to perform accurate offline experiments.
That is, given a dataset of interactions $\mathcal{D}_{0}$ that were logged under the production policies $\mathcal{G}_{0}$ and $\mathcal{R}_{0}$ (often referred to as \emph{logging} policies), we want to estimate what the reward \emph{would have been} if we had deployed a new policy $\mathcal{R}$ instead (assume $\mathcal{R}$ includes sampling candidates from $\mathcal{G}$).
The goal at hand is thus to devise some function $f$ that takes in a dataset of interactions collected under the logging policy, and is able to approximate the ground truth metric for a target policy $\mathcal{R}$, as shown in Eq.~\ref{eq:offline_exp}:
\begin{equation}\label{eq:offline_exp}
    \mathop{\mathbb{E}}\limits_{(a_{1},\ldots,a_{k}) \sim \mathcal{R}(x)}[C] \stackrel{?}{\approx} f(\mathcal{D}_{0},\mathcal{R}).
\end{equation}
\section{Discounted Cumulative Gain as an unbiased offline evaluation metric}\label{sec:dcg_unbiased}
In reality, this problem can be very complex.
Indeed, the reward that we obtain from presenting a certain ranking to a user can depend on the entire slate at once (of which there are $n!$ versions in a top-$n$ setting), and will be non-i.i.d. over trajectories (i.e. the reward distribution can depend on a user's state, influenced by actions we have taken in the past).

As is typical in machine learning research, we require assumptions that make the problem more tractable.
These assumptions are flawed, but they give us a starting point and a strong foundation to build upon for future iterations.
Note that we will lay out the \emph{specific} assumptions that are \emph{necessary} to motivate the use of Discounted Cumulative Gain as an offline evaluation metric---we discuss ways of relaxing these assumptions in Section~\ref{sec:beyond}.

\begin{assumption}[reward independence across trajectories]\label{ass:no_RL}
The reward for a context-action pair $(X,A_{i})$ in trajectory $\tau$ is independent of the rankings presented in other trajectories $\tau^{\prime} \in \mathcal{D}_{\setminus \tau}$.
\end{assumption}

\begin{assumption}[position-based model~\cite{Craswell2008}]\label{ass:pbm}
We follow the position-based model (PBM) to describe user scrolling behaviour, implying that the probability of a user viewing an item is only dependent on its rank and described by $\mathsf{P}(V|R)$.
\end{assumption}

\begin{assumption}[reward independence across ranks]\label{ass:no_slate}
The reward for a context-action pair $(X,A_{i})$ is independent of other actions in the user trajectory.
Formally, $C_{i} \perp \!\!\! \perp A_{j} | X, A_{i} \forall j \neq i$.
We describe the resulting reward distribution as $\mathsf{P}(C_{i}|X,A_{i},R)$.
\end{assumption}

\begin{assumption}[examination hypothesis~\cite{Craswell2008}]\label{ass:exam}
The reward for an action $A$ shown at rank $R$ is dependent on its inherent quality (unobserved random variable $Q$), and whether it was viewed $V$.
These quantities relate as:
\begin{equation}\label{eq:quality}
\begin{gathered}
    \mathsf{P}(C|X,A,R) = \mathsf{P}(Q|X,A) \cdot \mathsf{P}(V|R),\\
    {\text{ which implies }} \mathsf{P}(Q|X,A)=\frac{\mathsf{P}(C|X,A,R)}{\mathsf{P}(V|R)}.
\end{gathered}
\end{equation}
\end{assumption}
Asm.~\ref{ass:no_RL} allows us to avoid reinforcement learning scenarios, and Asm.~\ref{ass:pbm} prohibits cascading behaviour (that would give rise to other metrics, such as ERR~\cite{Chapelle2009}).
Asm.~\ref{ass:no_slate} allows us to avoid modelling entire slates as individual actions (leading to a combinatorial explosion of the action space), and through Asm.~\ref{ass:exam}, Eq.~\ref{eq:quality} estimates the unobserved context-dependent \emph{quality} of a given item from observable quantities alone.
From Eq.~\ref{eq:online_exp}, we can now rewrite the expected reward we obtain under a ranking policy $\mathcal{R}$ as:
\begin{flalign}\label{eq:dcg_subtle}
  \begin{aligned}
    \mathop{\mathbb{E}}\limits_{(a_{1},\ldots,a_{k}) \sim \mathcal{R}(x)}[C] &\approx \\
    \frac{1}{|\mathcal{D}_{0}|\cdot|\tau|}\sum_{\tau \in \mathcal{D}_{0}}\sum_{i=1}^{|\tau|}&\mathsf{P}(Q=1|X=x_{\tau},A=a_{i})\cdot\mathsf{P}(V=1|R=i).    
\end{aligned}
\end{flalign}
Through the position-based model and the examination hypothesis, items shown at lower ranks are \emph{discounted}.
Because we assumed independence of rewards across ranks, rewards observed at different ranks are \emph{cumulative}.
A key insight here is that the ranking policy $\mathcal{R}$ only affects the rank $i$ at which an item is shown, and as such, the exposure probability that is allocated to the item $a_{i}$.
We formalise \emph{exposure} as the expected number of views a target item $a^{\prime}$ will obtain under a given context $x$, for candidate generation and ranking policies $\mathcal{G},\mathcal{R}$:
\vspace{-3ex}
\begin{flalign}\label{eq:exposure}
\begin{gathered}
\mathop{\mathbb{E}}\limits_{\mathcal{G},\mathcal{R}}[V|X=x,A=a^{\prime}] = \sum\limits_{A^{k} \in \mathcal{A}^{k}} \Bigg( \mathcal{G}\left(A^{k}|X=x\right)\cdot\Bigg.\\
\Bigg.\hspace{-5ex}\sum\limits_{(a_{1},\ldots,a_{k}) \in S(A^{k})} \hspace{-5ex}\mathcal{R}\left((a_{1},\ldots,a_{k})|A^{k}, X=x\right) \sum\limits_{i=1}^{k} \mathsf{P}(V=1|R=i) \cdot\mathbb{1}_{\{a_i=a^{\prime}\}}\Bigg).
\vspace{-3ex}
\end{gathered}
\end{flalign}
Note that this general notation accommodates recommendation scenarios where we retrieve and rank $k$ candidates but only show the top-$n$ to the user, where $k>n$, by simply defining $\mathsf{P}(V=1|R=j)=0 \forall j=n+1,\ldots,k$.
By encoding cut-offs directly in the position bias model, we forgo the need to consider metrics like DCG@$n$~\cite{Valcarce2020}.

Recent work in unbiased learning-to-rank leverages similar exposure definitions to jointly combat selection and position bias through Inverse Propensity Score (IPS) weighting~\cite{Oosterhuis2020,Gupta2023}.
\begin{assumption}[full support of the logging policy~\cite{Owen2013}]\label{ass:full_support}
Given context $x$, any item that would be assigned non-zero exposure under the target policies $(\mathcal{G}, \mathcal{R})$ has non-zero exposure under the logging policies $(\mathcal{G}_{0}, \mathcal{R}_{0})$:
\begin{equation}
   \forall a \in \mathcal{A}: \mathop{\mathbb{E}}\limits_{\mathcal{G},\mathcal{R}}[V|X=x,A=a] > 0\Rightarrow \mathop{\mathbb{E}}\limits_{\mathcal{G}_{0},\mathcal{R}_{0}}[V|X=x,A=a] > 0.
\end{equation}
\end{assumption}
Through Asm.~\ref{ass:full_support} and Eq.~\ref{eq:exposure}, we can now formulate an \emph{importance sampling} estimator for the reward under $(\mathcal{G},\mathcal{R})$ given data collected under $(\mathcal{G}_{0}, \mathcal{R}_{0})$.
Let $\varepsilon(x,a)\coloneqq\mathbb{E}_{\mathcal{G},\mathcal{R}}[V|X=x,A=a]$ and $\varepsilon_{0}(x,a)\coloneqq\mathbb{E}_{\mathcal{G}_{0},\mathcal{R}_{0}}[V|X=x,A=a]$.
\begin{equation}\label{eq:unbiased_reward}
\mathop{\mathbb{E}}\limits_{\mathcal{G},\mathcal{R}}[C] = \mathop{\mathbb{E}}\limits_{\mathcal{G}_{0},\mathcal{R}_{0}}\left[C \cdot \frac{\varepsilon}{\varepsilon_{0}}\right] \approx \frac{1}{|\mathcal{D}_{0}|\cdot|\tau|}\sum_{\tau \in \mathcal{D}_{0}}\sum_{i=1}^{|\tau|} c_{i} \cdot \frac{\varepsilon(x_{\tau},a_{i})}{\varepsilon_{0}(x_{\tau},a_{i})}.
\end{equation}
Eq.~\ref{eq:unbiased_reward} provides a general unbiased estimator for the reward under $(\mathcal{G}, \mathcal{R})$, computed from data collected under $(\mathcal{G}_{0}, \mathcal{R}_{0})$.\footnote{Note that, for general two-stage ranking scenarios, this is a novel contribution to the research literature in and of itself. Existing work on off-policy corrections in two-stage recommender systems only considers a top-$1$ scenario, instead of a ranking policy~\cite{Ma2020}.}
For simplicity of notation, but without loss of generality, we now restrict ourselves to a deterministic ranking policy $\mathcal{R}$ and assume $\mathcal{G}$ is constant (i.e. $\mathcal{G}\equiv\mathcal{G}_{0}$).
With a slight abuse of notation, we briefly denote with $\mathcal{R}(x,a)$ the \emph{rank} at which item $a$ is placed when policy $\mathcal{R}$ is presented with context $x$.
In doing so, we observe that the importance weights in Eq.~\ref{eq:unbiased_reward} can be simplified to:
$ 
    \frac{\varepsilon(x,a)}{\varepsilon_{0}(x,a)}=\frac{\mathsf{P}(V=1|R=\mathcal{R}(x,a))}{\mathsf{P}(V=1|R=\mathcal{R}_{0}(x,a))}.$  
We stress again that we simply adopt this view for simplified notation, but that the derivation holds for general \emph{stochastic} two- or single-stage ranking systems alike.

Recall from Eq.~\ref{eq:quality} that $\mathsf{P}(Q|X,A)=\frac{\mathsf{P}(C|X,A,R)}{\mathsf{P}(V|R)}$.

Now, we can formally describe the \emph{discounted cumulative gain} (DCG) metric as an importance sampling estimator:
\begin{flalign}\label{eq:formal_DCG}
\begin{aligned}
    &\mathop{\mathbb{E}}\limits_{(a_{1},\ldots,a_{k}) \sim \mathcal{R}(x)}[C] \approx f_{\rm DCG}(\mathcal{D}_{0},\mathcal{R})\\
     &= \frac{1}{|\mathcal{D}_{0}|\cdot|\tau|}\sum_{\tau \in \mathcal{D}_{0}}\sum_{i=1}^{|\tau|} c_{i}\cdot\frac{\mathsf{P}\left(V=1|R=\mathcal{R}(x_{i},a_{i})\right)}{\mathsf{P}(V=1|R=i)} \\
     &\approx \frac{1}{|\mathcal{D}_{0}|\cdot|\tau|}\sum_{\tau \in \mathcal{D}_{0}}\sum_{i=1}^{|\tau|}  \mathsf{P}(Q|X=x_{\tau},A=a_{i})\cdot\mathsf{P}(V=1|R=\mathcal{R}(x_{i},a_{i})).    
\end{aligned}
\end{flalign}
Through this derivation, two different views of the DCG metric arise.
That is, we either
\begin{enumerate*}
\item view it as a pure importance sampling estimator that reweights the exposure that is allocated to a certain item in a certain context, or
\item we view it as a way to de-bias observed interactions (i.e. estimate $Q$ from $C$ and $V$), and use the position-based model (Asm.~\ref{ass:pbm}) and the examination hypothesis (Asm.~\ref{ass:exam}) to obtain a final estimate of the cumulative reward.
\end{enumerate*}

If the assumptions laid out above hold, Eq.~\ref{eq:formal_DCG} provides an unbiased estimate of the online reward policy $\mathcal{R}$ will incur, based on data collected under $\mathcal{R}_{0}$; providing a strong motivation for DCG.
Even though unbiasedness is an attractive theoretical property, this estimator's variance can become problematic in cases where the logging and target policies $(\mathcal{R}_{0},\mathcal{R})$ diverge. 
We can adopt methods that were originally proposed to strike a balance between bias and variance for general IPS-based estimators, such as clipping the weights~\cite{Ionides2008,Gilotte2018}, self-normalising them~\cite{Swaminathan2015snips}, adapting the logging policy~\cite{Tucker2023}, or extending the estimator with a reward model to enable doubly robust estimation~\cite{Dudik2011,Kiyohara2022,Oosterhuis2023}.
Similarly, when the logged data that is used for offline evaluation was collected by multiple logging policies, ideas from ``\emph{multiple importance sampling}''~\cite{Elvira2019} are effective at reducing the variance of the final estimator~\cite{Agarwal2017,Kallus2021Optimal}.
\citeauthor{Saito2023_ICML} describe extensions for large action spaces~\cite{Saito2022_ICML,Saito2023_ICML}.

In the traditional IR use-case of web search, it is often assumed that we have access to human-annotated relevance labels ${\rm rel}(q,d)$ for query-document pairs $(q,d)$.
Such crowdsourced labels are seen as a proxy to $\mathsf{P}(Q|X,A)$, which makes them understandably attractive.
Nevertheless, for an offline evaluation metric to be useful in real-world recommendation systems, access to direct relevance labels is seldom a realistic requirement.
The \emph{discount} function for DCG that is most often used in practice, makes the assumption that $\mathsf{P}(V=1|R=i)=\frac{1}{\log_{2}(i+1)}$ is a good approximation for empirical exposure (see, e.g.,~\cite{Chapelle2009}).
This gives rise to the more widely recognisable form of DCG, as $\sum_{i=1}^{k}\frac{{\rm rel}(q,d_{i})}{\log_{2}(i+1)}$.
We note that neither one of these additional assumptions is likely to hold in real-world applications, which imply that even if Assumptions~\ref{ass:no_RL}--\ref{ass:full_support} hold, this estimator is \emph{biased}.
The more general form presented in Eq.~\ref{eq:formal_DCG}, however, is supported by a theoretical framework that allows for counterfactual evaluation, formally describing settings for its appropriate use.
\begin{table*}[!t]
\setlength{\fboxrule}{1pt}
\begin{minipage}{0.65\textwidth}
\begin{flushleft}
    \begin{tabular}{cccccccc}
    \toprule
    \textbf{Top-1 Model} & \textbf{DCG}($x_{1}$) & \textbf{DCG}($x_{2}$) & \textbf{nDCG}($x_{1}$) & \textbf{nDCG}($x_{2}$) &~& \textbf{DCG}($\mathbf{X}$) & \textbf{nDCG}($\mathbf{X}$) \\
    \midrule
    $\mathcal{R}(x) = a_{1}$ & 1.00 & 1.00 & 1.00 & 0.29 &~& 1.00 & \fcolorbox{Maroon}{white}{0.64}\\
    $\mathcal{R}^{\prime}(x) = a_{2}$ & 0.00 & 2.50 & 0.00 & 0.71 &~& \fcolorbox{PineGreen}{white}{1.25} & 0.36\\
    \bottomrule
    \end{tabular}
\end{flushleft}
\end{minipage}
\hspace{2ex}
\begin{minipage}{0.32\textwidth}
\begin{flushright}
    ~\\
where $\quad$ $\mathbb{E}[Q|X=x_{1},A=a_{1}] = 1.0$,\\
    $\mathbb{E}[Q|X=x_{1},A=a_{2}] = 0.0$,\\
    $\mathbb{E}[Q|X=x_{2},A=a_{1}] = 1.0$,\\
    $\mathbb{E}[Q|X=x_{2},A=a_{2}] = 2.5$,\\
    ~\hfill~$\textbf{X} = \{x_{1}, x_{2}\}$.
\end{flushright}
\end{minipage}
\caption{A proof by example that, while rankings inferred from the DCG and nDCG metrics are consistent for a \emph{single sample}, they can be \emph{inconsistent} when aggregated over multiple samples (i.e. ${\rm DCG}(\bm{X}, \mathcal{R}^{\prime}) > {\rm DCG}(\bm{X},\mathcal{R}) \nRightarrow {\rm nDCG}(\bm{X},\mathcal{R}^{\prime}) > {\rm nDCG}(\bm{X},\mathcal{R}) $).}\label{tab:proof}
\end{table*}

\section{Normalising DCG is Inconsistent}\label{sec:ndcg_inconsistent}
We have introduced the DCG metric from first principles, and have shown that under several assumptions, it can be seen as an unbiased estimator for the reward that a new ranking policy $\mathcal{R}$ will obtain. 
Nevertheless, this metric seldom appears in the research literature in its unadulterated form.
Much more prevalent is \emph{normalised} DCG (nDCG), which rescales the DCG metric to be at most $1$ for an ``ideal'' ranking.
For notational simplicity, we define ground truth quality labels as $\rho(x) = \left[ \mathbb{E}[Q|X=x,A=a] \forall a \in \mathcal{A}\right]$.
Then, we can compute \emph{ideal} DCG as the DCG obtained under an oracle ranker that yields $\mathcal{R}^{\star}\coloneqq{\texttt{sort}}(-\rho(x))$.
In what follows, we first show that this is reasonable practice under a single sample (context), in that it retains a \emph{consistent} ordering among ranking policies.
We then go on to show that:
\begin{enumerate*}
    \item nDCG yields \emph{inconsistent} orderings over competing policies in expectation when compared to DCG, and
    \item defining iDCG is problematic in the realistic setting of partial information (i.e. $Q$ and thus $\rho(x)$ are unobservable).
\end{enumerate*}

\begin{lemma}\label{lem:consistent_single}
The Discounted Cumulative Gain (DCG) and Normalised Discounted Cumulative Gain (nDCG) metrics yield consistent relative orders over a competing set of policies $\Omega$ that are being evaluated \textbf{for a single sample $\bm{x}$}.
That is, $$\mathop{\argsort}\limits_{\mathcal{R} \in \Omega} f_{\rm DCG}(x,\mathcal{R}) \equiv \mathop{\argsort}\limits_{\mathcal{R} \in \Omega} f_{\rm nDCG}(x,\mathcal{R}), \forall x \in \mathcal{X}.$$
\end{lemma}
\begin{proof}
For any given context $x \in \mathcal{X}$, assume a relative order exists between two methods $\mathcal{R}$ and $\mathcal{R}^{\prime}$ for a given non-negative metric $f$: i.e. $f(x,\mathcal{R}) \geq f(x,\mathcal{R}^{\prime})$.
Define $f^{\star}(x)$ as the \emph{ideal} metric value, i.e. the metric value that is obtained by the optimal ranking:
$f^{\star}(x) \coloneqq f(x, \mathcal{R}^{\star})$, where $\mathcal{R}^{\star} = \mathop{\argmax}\limits_{\mathcal{R}} f(x, \mathcal{R}) = {\texttt{sort}}(-\rho(x))$.

Because $f$ is a non-negative metric, $f^{\star}$ is non-negative, and we have that:
$$ f(x,\mathcal{R}) \geq f(x,\mathcal{R}^{\prime}) \Rightarrow  \frac{f(x,\mathcal{R})}{f^{\star}(x)} \geq \frac{f(x,\mathcal{R}^{\prime})}{f^{\star}(x)}. $$

DCG is, in general, not restricted to be non-negative.
However, if we assume that the \emph{ideal} discounted cumulative gain is non-negative, we have that the above inequality applies for $f \coloneqq f_{\rm DCG}$ and $f^{\star} \coloneqq f_{\rm iDCG}$.
\end{proof}

We believe that this (seemingly trivial) insight has led to the widespread adoption of nDCG as an offline evaluation metric in the recommender systems and Learning-to-Rank research fields, as normalised metric values where $1$ indicates a perfect model facilitate comparisons of methods over different datasets.
Indeed, when describing its prevalence in the literature, \citeauthor{Ferrante2021} argue that ``\emph{usually nDCG is preferred over DCG because it is bounded and normalised}''~\cite{Ferrante2021}.
Nevertheless, nDCG does \emph{not} retain consistent orderings with respect to DCG when the metrics are calculated over multiple samples and aggregated:

\begin{lemma}
The Discounted Cumulative Gain (DCG) and Normalised Discounted Cumulative Gain (nDCG) metrics yield \textbf{inconsistent} relative orders over a competing set of policies $\Omega$ that are being evaluated \textbf{over a set of samples $\bm{X}$}.
That is, $$\mathop{\argsort}\limits_{\mathcal{R} \in \Omega} f_{\rm DCG}(\bm{X},\mathcal{R}) \nequiv \mathop{\argsort}\limits_{\mathcal{R} \in \Omega} f_{\rm nDCG}(\bm{X},\mathcal{R}) \text{ in the general case}.$$
\end{lemma}

\begin{proof}
    We provide a proof by counterexample, for which the details are presented in Table~\ref{tab:proof}.
    Indeed, even though the $f_{\rm DCG}$ and $f_{\rm nDCG}$ metrics align for every sample in \emph{isolation} (see columns for $x_1, x_2$), they are inconsistent in \emph{aggregate} (see columns for $\bm{X}$).
\end{proof}

Discrepancies between (n)DCG have been touched upon in the IR literature, focused on search engine evaluation and blaming ``\emph{a limited number of relevance judgments}''~\cite{AlMaskari2007}.
Table~\ref{tab:proof} shows that the issue has deeper roots than this: the normalisation procedure is \emph{inconsistent}.
This insight is problematic, as virtually all offline evaluation protocols consist of first aggregating evaluation metrics over sets of samples, and then inferring preferences over competing policies based on these metric averages.
Our work shows that, when the assumptions laid out in Section~\ref{sec:dcg_unbiased} are met and DCG provides an unbiased estimate of reward, this gives rise to a theoretically sound model selection protocol.
The same statement does \emph{not} hold for nDCG as it widely appears in the research literature (see e.g. \cite[Eq. 8.9]{schutze2008introduction} and \cite[Eq. 5]{Valcarce2020}).
We hypothesise that the normalisation formula has been widely adopted for its ease-of-use, rather than for its theoretical properties.
This suggests that nDCG is of limited practical use as an offline evaluation metric, and that it should be avoided by researchers and practitioners who wish to use DCG for offline evaluation and model selection purposes.

We provide empirical evidence of discrepancies between (n)DCG on common top-$n$ recommendation evaluation tasks on publicly available data in Appendix~\ref{sec:appx}.

In the odd case where metric values that maximise at $1$ are required, we propose the use of a \emph{post}-normalisation procedure, where $f_{\rm pnDCG}(\mathcal{D},\mathcal{R}) = \frac{f_{\rm DCG}(\mathcal{D},\mathcal{R})}{f_{\rm iDCG}(\mathcal{D})}$.
Recent work on LTR in IR leverages this nDCG formulation~\cite{Oosterhuis2022}.
Indeed, through Lemma~\ref{lem:consistent_single}, one can trivially show that this metric \emph{is} consistent with respect to $f_{\rm DCG}$.

Nevertheless, we wish to advise against this practice altogether, as computing the \emph{ideal} DCG metric implies that we must construct the \emph{ideal} ranking policy $\mathcal{R}^{\star}$.
To do so, we require full knowledge of $\rho(x)$, which is hardly realistic in real-world scenarios.
In traditional IR use-cases where human-annotated relevance labels are available as \emph{ground truth}, these labels can be used to inform $\rho(x)$.
In academic recommendation datasets where we have \emph{explicit} feedback and \emph{full observability} of the user-item matrix, this can inform $\rho(x)$ similarly.
In practical applications, however, we typically estimate $\widehat{\rho}(x) \approx \rho(x)$ from logged implicit feedback.
Aside from the problems that occur when accurately interpreting this feedback~\cite{Diaz2021}, such logged datasets are known to be riddled with biases that complicate estimating $\widehat{\rho}(x)$~\cite{Jeunen2021A}, resulting in only partial observability and noisy estimates that should \emph{not} be taken at face value to inform an ``\emph{optimal}'' ranker.\looseness=-1

One final argument in favour of nDCG, is that it partially alleviates the impact of outliers.
Indeed, the normalisation procedure rescales the contribution of every sample, which can be preferable in cases where strong outliers are present.
Nevertheless, in such scenarios, we would propose to \emph{first} devise a more appropriate online metric than the average cumulative reward, and \emph{then} derive an offline estimator for this quantity, rather than trying to repurpose the existing DCG estimator.
Exactly what such metrics and estimators would look like, is an interesting area for future work.

\section{Experimental Results \& Discussion}\label{sec:experiments}
Until now, we have derived the theoretical conditions that are necessary to consider DCG an unbiased estimator of online reward, and we have highlighted both theoretically and empirically that nDCG deviates from this.
In what follows, we wish to empirically validate whether we can leverage the metrics effectively to estimate online reward for deployed ranking policies, for recommender systems running on large-scale web platforms.
The research questions we wish to answer with empirical evidence, are the following:
\begin{description}
    \item[\textbf{RQ1}] \textit{Are offline evaluation results obtained using DCG correlated with online metrics from deployed models?}
    \item[\textbf{RQ2}] \textit{Are offline evaluation results obtained using normalised DCG correlated with online metrics from deployed models?}
    \item[\textbf{RQ3}] \textit{Does the proposed unbiased DCG formulation with learnt position biases and de-biased interaction labels improve correlation with online metrics over the classical and widely adopted (yet biased) DCG formulation?}
    \item[\textbf{RQ4}] \textit{Are differences in any of the considered offline evaluation metrics predictive of differences in online metrics?}
\end{description}
We focus on correlations between off- and online metrics rather than exact estimation error, because downstream business logic prevents the model output to \emph{exactly} match online behaviour~\cite{Jakimov2023}.

To provide empirical evidence for research questions 1--4, we require access to a \emph{ground truth} online metric.
There are two families of evaluation methods we can consider to obtain this:
\begin{enumerate*}
    \item simulation studies, or
    \item online experiments.
\end{enumerate*}
Both come with their own (dis)advantages.
Indeed, simulation studies are generally reproducible and allow full control over the environment to investigate which of the assumptions laid out in the theoretical sections of this work are necessary to retain DCG's utility as an online metric.
Nevertheless, simulations require us to make additional assumptions about user behaviour, that are often non-trivial to validate.
As a result, they would provide no empirical evidence on real-world value of the metrics, and are limited in the insights they can bring.

Online experiments, on the other hand, are harder to reproduce.
Notwithstanding this, they allow us to directly measure real user behaviour and give a view of the utility of the DCG metric for offline evaluation purposes in a real-world deployed system.
This can guide practitioners who need to perform such evaluations.
We focus on the online family for the remainder of this work, noting that simulations provide an interesting avenue for future research.

We use data from a large-scale social media platform that utilises a two-stage ranking system as described earlier to present users with a personalised feed of short videos they might enjoy.
The platform operates a hierarchical feed where users are presented with a 1\textsuperscript{st}-level feed they can scroll through and engage with content, and users can enter a 2\textsuperscript{nd}-level ``\emph{more-like-this}'' feed via any given 1\textsuperscript{st}-level item.\footnote{Examples of well-known social media platforms that operate similar user interfaces include Instagram, Reddit, and ShareChat.}
Because of the differences in interaction modalities and the user interface between the two feeds, they require separate models to estimate position bias, and we separate them in our analysis.
The 1\textsuperscript{st}-level feed adopts a recently proposed probabilistic position bias model~\cite{Jeunen2023_C3PO}, whereas the 2\textsuperscript{nd}-level feed adopts an exponential form (such as the one underlying rank-biased precision~\cite{Moffat2008}).
Because of this difference, the importance weights in the 2\textsuperscript{nd}-level feed exhibit much larger variance, and we adopt a clipping parameter for IPS which we set at $200$ to compute the de-biased DCG metric on this data (and vary it in Section~\ref{sec:RQ3}).
\emph{Rewards} on a short-video platform can be diverse.
We collect both implicit signals (e.g. watching a video until the end) and explicit signals (e.g. liking a video), and consider both types of rewards for our on- and offline metrics, referring to them as $\mathcal{C}_{\rm imp}$ and $\mathcal{C}_{\rm exp}$ respectively.
\sisetup{detect-weight=true,detect-inline-weight=math}
\begin{table*}[t]
    \centering
    \begin{tabular}{cccRpcRpcRpcRpccRp}
    \toprule
     ~&~& ~ & \multicolumn{11}{c}{\textbf{DCG}} &~&~& \multicolumn{2}{c}{\textbf{nDCG}} \\
    \cmidrule(l){4-14}\cmidrule(l){17-18}
     ~&~ & ~ & \multicolumn{2}{c}{\textbf{log}, $\mathbf{C}$} &~& \multicolumn{2}{c}{\textbf{log}, $\mathbf{\hat{Q}}$}  &~& \multicolumn{2}{c}{\textbf{pbm},  $\mathbf{C}$} &~& \multicolumn{2}{c}{\textbf{pbm}, $\mathbf{\hat{Q}}$}  &~&~& \multicolumn{2}{c}{\textbf{pbm}, $\mathbf{\hat{Q}}$} \\
    \cmidrule(l){4-5}\cmidrule(l){7-8}\cmidrule(l){10-11}\cmidrule(l){13-14}\cmidrule(l){17-18}
\rowcolor{White} ~& ~ & ~ & $\bm{r}$ & $\bm{p}$\textbf{-val} &~& $\bm{r}$ & $\bm{p}$\textbf{-val} &~& $\bm{r}$ & $\bm{p}$\textbf{-val} &~& $\bm{r}$ & $\bm{p}$\textbf{-val}  &~&~& $\bm{r}$ & $\bm{p}$\textbf{-val} \\
    \midrule
\multirow{4}{*}{1\textsuperscript{st}-level} &
 \multirow{2}{*}{$C_{\rm exp}$} & $\mathcal{R}_{0}$ & 0.97\textsuperscript{\dag} & <0.01 &~& 0.96\textsuperscript{\dag} & <0.01 &~& 0.97\textsuperscript{\dag} & <0.01 &~& \bfseries0.98\textsuperscript{\dag} & <0.01 &~&~& -0.91\textsuperscript{\dag} & 0.01\\
~& ~ & $\mathcal{R}_{t}$ & 0.86\textsuperscript{\dag}  & 0.03 &~& 0.86\textsuperscript{\dag} & 0.03 &~& 0.87\textsuperscript{\dag} & 0.03 &~& \bfseries0.91\textsuperscript{\dag} & 0.01 &~&~& -0.80 & 0.06\\
\vspace{-2ex}~&~&~&~&~&~&~&~&~&~&~&~&~&~&~&~&~&\\
~& \multirow{2}{*}{$C_{\rm imp}$} & $\mathcal{R}_{0}$ & 0.57 & 0.24 &~& \bfseries0.70 & 0.12 &~& 0.57 & 0.24 &~& \bfseries0.70 & 0.12 &~&~& -0.80 & 0.06\\
~& ~ & $\mathcal{R}_{t}$ & 0.05 & 0.92 &~& \bfseries0.17 & 0.73 &~& 0.05 & 0.92 &~& \bfseries0.17 & 0.75 &~&~& -0.32 & 0.55\\
\vspace{-2ex}~&~&~&~&~&~&~&~&~&~&~&~&~&~&~&~&~&\\
\cdashline{1-18}\\
\vspace{-5ex}~&~&~&~&~&~&~&~&~&~&~&~&~&~&~&~&~&\\
\multirow{4}{*}{2\textsuperscript{nd}-level} &
\multirow{2}{*}{$C_{\rm exp}$} & $\mathcal{R}_{0}$ & 0.75 & 0.14 &~& 0.71 & 0.18 &~& 0.77 & 0.13 &~& \bfseries0.79 & 0.11 &~&~& -0.42 & 0.48\\
 ~&~ & $\mathcal{R}_{t}$ & 0.28 & 0.65 &~& 0.25 & 0.68 &~& 0.38 & 0.53 &~& \bfseries0.49 & 0.40 &~&~& -0.34 & 0.57\\
\vspace{-2ex}~&~&~&~&~&~&~&~&~&~&~&~&~&~&~&~&~&\\
~& \multirow{2}{*}{$C_{\rm imp}$} & $\mathcal{R}_{0}$ & 0.10 & 0.88 &~& 0.58 & 0.30 &~& 0.48 & 0.41 &~& \bfseries0.85 & 0.07 &~&~& -0.93\textsuperscript{\dag} & 0.02\\
~& ~ & $\mathcal{R}_{t}$ & 0.55 & 0.34 &~& \bfseries0.88\textsuperscript{\dag} & 0.05 &~& 0.70 & 0.19 &~& \bfseries0.88\textsuperscript{\dag} & 0.05 &~&~& -0.77 & 0.13\\
    \bottomrule
    \end{tabular}
\caption{Correlation between online reward as measured from an A/B-test and offline evaluation metrics. We consider both \emph{explicit} and \emph{implicit} reward signals, on two levels of a hierarchical feed structure on a short-video platform. We consider DCG with a logarithmic discount (log) and a learnt model (pbm); using interaction signals directly ($C$) or de-biasing them to estimate $\hat{Q}=C/\mathsf{P}(V|R)$; for both a logging and target policy $\mathcal{R}_{0},\mathcal{R}_{t}$.
We report Pearson's correlation coefficient $r$ and a two-tailed $p$-value using Student's correlation test~\cite{Student1908}.
Statistically significant correlations ($p<0.05$) are marked\textsuperscript{\dag}, best performers are bold.}~\label{tab:results}
\end{table*}
\subsection{Offline--Online Metric Correlation (RQ1--3)}\label{sec:RQ1}
We collect data over a week of a deployed online experiment with over 40 million users where we deployed a change to a deterministic ranking policy $\mathcal{R}$, and kept the candidate generator $\mathcal{G}$ fixed.
This simplifies the exposure calculation that is used in the offline evaluation metrics in Eq.~\ref{eq:exposure} to that in Eq.~\ref{eq:formal_DCG}, and ensures that the variance of the offline estimator is lower.
For every day $d$ in the experiment, we
\begin{enumerate*}
    \item aggregate online results per day $d$ as the average number of logged positive feedback samples per session, and
    \item collect a dataset $\mathcal{D}^{d}_{0}$ which we use to compute offline metrics (through Eq.~\ref{eq:formal_DCG} and variants thereof).
\end{enumerate*}
Then, we compute Pearson's correlation coefficient between the series of $d$ online metrics from (1), and the $d$ offline estimates from (2), for competing evaluation metrics.
Table~\ref{tab:results} presents results from this experiment, with a detailed description in the caption.
Results here align with what theory would suggest: the unbiased DCG variant that we have formally derived in Section~\ref{sec:dcg_unbiased} provides the strongest correlation with online reward.
Both adopting a \emph{learnt} position bias model as opposed to the classical logarithmic form, and de-biasing observed interaction labels as opposed to na\"ively using them, have a significant effect on the performance of the final estimator.
On the 1\textsuperscript{st}-level feed, we observe that our estimator works especially well with an explicit reward signal.
This somewhat reverses for the 2\textsuperscript{nd}-level feed, where the implicit reward signal leads to a stronger correlation with online results.
We hypothesise that this is directly related to the degree with which the assumptions laid out in Section~\ref{sec:dcg_unbiased} are violated.
Users leave the 1\textsuperscript{st}-level feed and enter the 2\textsuperscript{nd}-level feed if they click on a 1\textsuperscript{st}-level video.
As such, the implicit reward $\mathcal{C}_{\rm imp}$ of succesfully watching a video is highly dependent on other items in the feed, violating Asm.~\ref{ass:no_slate} (reward independence across ranks), as well as Asm.~\ref{ass:pbm} (absence of cascading behaviour).
Note that all assumptions are expected to be violated to some degree --- but the strongly positive correlation results presented in Table~\ref{tab:results} are promising.

Somewhat surprisingly, not only does (unbiased) normalised DCG exhibit worse correlation than DCG, it provides a strongly \emph{negative} correlation.
At first sight, this seems troubling.
Nevertheless, we provide an intuitive explanation and highlight that this result alone does \emph{not} imply that the metric cannot be useful for offline evaluation purposes.
Note that discrepancies stemming from the normalisation procedure have a disproportionate impact when the reward is unevenly distributed across sessions or days.
Suppose we observe two sessions: one with a single positive label over two impressions, and one with 10 positive labels over 1\,000 impressions.
Because DCG deals with absolute numbers, the second session will bear a weight proportional to its positive reward.
Normalised DCG, on the other hand, considers the relative distance to the optimal ranking.
If we reasonably assume that our ranking model is imperfect, the distance to the optimal ranking is likely to be \emph{higher} for the second session, and nDCG will be lower as a result (even though we have higher DCG).
This same argument can be made across different days in the experiment, explaining poor correlation results over time.
Notwithstanding this, it does not necessarily imply that nDCG holds no merit as an offline evaluation metric.
In what follows, we consider a more important question, focusing on \emph{differences} in online metrics instead.

\subsection{Offline--Online Metric Sensitivity (RQ4)}\label{sec:RQ3}
To consider this research question, we restrict ourselves to a setting where we know that strong statistically significant ($p\ll 0.001$) improvements in online metrics are observed for the target policies $\{\mathcal{G}_{t},\mathcal{R}_{t}\}$ over the logging policies $\{\mathcal{G}_{0},\mathcal{R}_{0}\}$.
We restrict our analysis to the 2\textsuperscript{nd}-level feed, as it generates the majority of user-item interactions and relatively long sessions. 
We consider a variety of explicit and implicit feedback signals as rewards, all of which saw statistically significant improvements in the online experiment.
\begin{figure*}[!t]
    \centering
    \begin{subfigure}[t]{0.48\textwidth}
        \centering
        \includegraphics[width=.95\textwidth, trim = {0cm 5mm 0cm 5mm}]{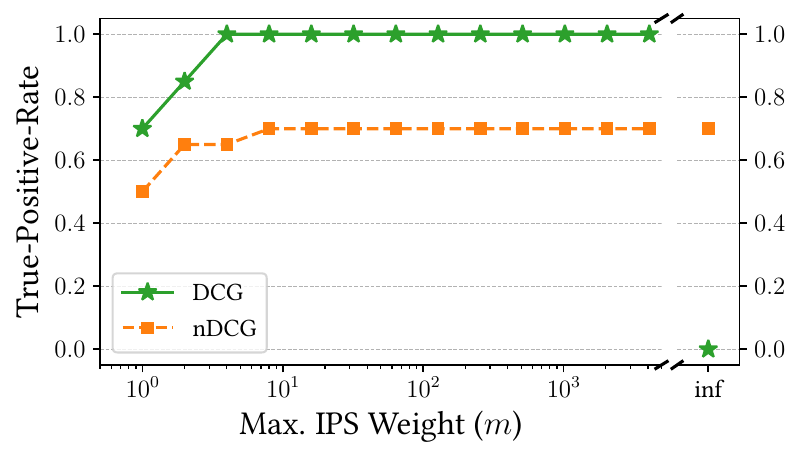}
        \caption{True-Positive-Rate for offline metrics ($\uparrow$).}\label{fig:sensitivity_TPR}
    \end{subfigure}%
    ~ 
    \begin{subfigure}[t]{0.48\textwidth}
        \centering
        \includegraphics[width=1.01\textwidth, trim = {0cm 5mm 0cm 5mm}]{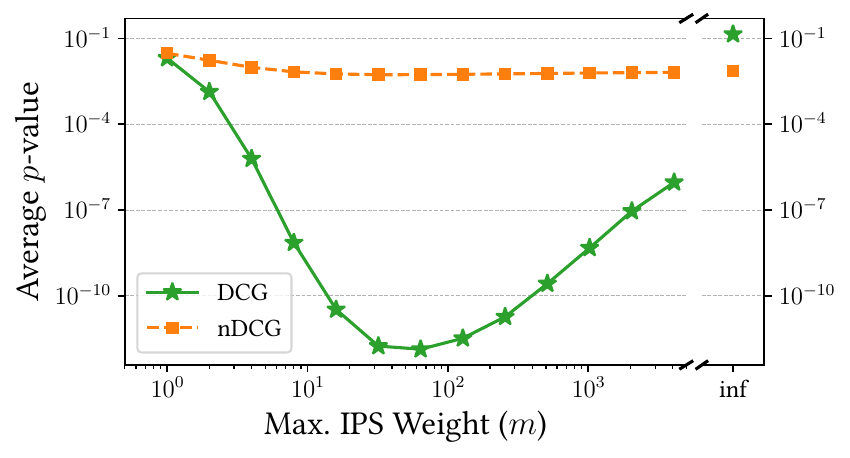}
        \caption{Average $p$-values for offline metrics ($\downarrow$).}\label{fig:sensitivity_P}
    \end{subfigure}
    \caption{Sensitivity measures (\emph{y-axis}) of (n)DCG for varying values of the capping parameter in IPS (\emph{x-axis}).}\label{fig:sensitivity_results}
 \Description[Figure described in text.]{Figure described in text.}
\end{figure*}
We consider three possible notions of \emph{alignment} between on- and offline metrics, in increasing levels of expressivity: 
\begin{enumerate*}
    \item Without considering statistical significance, do differences in offline metrics directionally align with differences in online metrics? (i.e. \emph{sign agreement}) 
    \item For statistically significant improvements (as validated by the online experiment with $p\ll 0.001$), does the offline metric show statistically significant improvements with $p < 0.01$? (i.e. \emph{True-Positive-Rate}, recall or sensitivity).
    \item For statistically significant improvements, what confidence does the offline metric have in the improvements? (i.e. the $p$-values).
\end{enumerate*}
We consider 5 days of a deployed online experiment and 4 different reward signals, yielding 20 distinct statistically significant online metric improvements.
The purpose of our offline estimators, is to reflect these statistically significant \emph{online} differences in their \emph{offline} estimates. 
As offline estimators, we consider a clipped variant of the unbiased DCG metric in Eq.~\ref{eq:unbiased_reward}, where the factor for the inverse exposure propensity is replaced with $\min\left(m, \frac{1}{\varepsilon_{0}}\right)$, for varying values of $m$~\cite{Ionides2008,Gilotte2018}.
Although this renders the metric \emph{biased} in a pessimistic way~\cite{Jeunen2023Pessimism}, its reduced variance can yield more favourable performance as an offline evaluation metric.
For every variant of the DCG metric we construct in this way, we compute the analogous \emph{normalised} DCG metric.
Because these metrics are aggregated over trajectories, we can use their empirical means and standard deviations to construct normal confidence intervals for the metric values and their differences (i.e. the treatment effect).
If the 99\% confidence interval for the metric difference is strictly positive, we say the metric indicates a statistically significant improvement (with $p <0.01$).
Because all online metric differences we consider were statistically significant, we know that they are \emph{true positives}, and we can compute the \emph{sensitivity} or \emph{True-Positive-Rate} (TPR) for the offline metric by counting how often it indicates a statistically significant \emph{offline} difference for these true positives.
We additionally record the average $p$-value for the null hypothesis (i.e. the hypothesis that the metric difference is $\leq 0$), obtained from the confidence intervals.
To measure the weaker notion of \emph{sign agreement}, we only consider the mode of the confidence interval $\mu$, and count ``agreement'' ${\rm iff }~\mu > 0$.
We vary the clipping hyper-parameter for IPS as $m \in \{1, 2, 4, 8, 16, 32, 64, 128, 256, 512, 2048, 4096, \inf\}$, where $m=1$ corresponds to directly using the interaction labels $C$, and $m=\inf$ yields an unbiased estimator with higher variance.

Reassuringly, \emph{all} offline estimators exhibit 100\% directional sign agreement with the true treatment effect we observe for serving personalised recommendations to users through $\{\mathcal{G}_{t},\mathcal{R}_{t}\}$ over $\{\mathcal{G}_{0},\mathcal{R}_{0}\}$.
Results for our sensitivity analysis are visualised in Figure~\ref{fig:sensitivity_results}.
In Figure~\ref{fig:sensitivity_TPR}, higher values indicate that the offline evaluation metric is more likely to detect statistically significant improvements in the online metric, averaged over the 20 settings described above.
Analogously, lower values in Figure~\ref{fig:sensitivity_P} indicate that the metric yields more statistical confidence.
Lower $p$-values in Figure~\ref{fig:sensitivity_P} additionally imply that the metric requires less data to achieve the significance level, potentially reducing costs~\cite{Kharitonov2017}. 
We observe that introducing IPS weighting (i.e. $m>1$) to account for position bias in the logged interactions leads to improved sensitivity.
This results holds for both DCG and nDCG, and both for the TPR metric and the average $p$-values.
We additionally observe that for a wide range of clipping values, the DCG metric has a higher TPR (lower $p$-values) than nDCG.
Intuitively, this can be explained by the fact that nDCG essentially squashes a metric with a high expressive range to the $[0,1]$ domain, which can only come at a cost of discriminative power.
DCG, on the other hand, directly models the online metrics we care about (under the assumptions laid out in Section~\ref{sec:dcg_unbiased}).

When we do not clip the IPS weights (i.e. $m=\inf$), we observe from Figure~\ref{fig:sensitivity_results} that the variance of the DCG metric increases to a point where its sensitivity is harmed (even if directional alignment is maintained).
Note that the nDCG metric is not affected by this, as its values are bounded and they exhibit lower variance as a result.
Intuitively, whereas low values of $\varepsilon_{0}$ can blow up the unbiased DCG formulation, they will also do this for \emph{ideal} DCG, and their ratio (i.e. nDCG) will be less likely to suffer from this.
We observe that with clipped propensities, even at large values, DCG leads to superior sensitivity over nDCG, striking a favourable bias-variance trade-off.

Our experimental results show promise in using DCG as an offline estimator of online reward, and the bias-variance trade-off that is clearly visualised in Figure~\ref{fig:sensitivity_P} helps us to tune this hyper-parameter $m$ properly.
Even when \emph{all} of the metric's underlying assumptions violated to \emph{some} degree, its value is apparent.
\vspace{-2ex}\section{Perspectives going forward}\label{sec:beyond}
In what follows, we revisit the assumptions that are \emph{necessary} to consider the DCG metric an unbiased estimator of online reward.

\textit{1. Reward independence across trajectories} is necessary to avoid having to model any internal user \emph{state} that is influenced by actions taken by a ranking policy.
Indeed, if we do allow this to happen, we must resort to Reinforcement Learning (RL) formulations of our problem, which inhibits the simple form that DCG allows.
Nevertheless, unbiased evaluation of RL policies is an active research area, which has found applications in recommendation research~\cite{Chen2019}, also for two-stage policies (without considering rankings)~\cite{Ma2020}.
\citeauthor{Ie2019SlateQ} can provide inspiration for learnt RL policies in top-$n$ recommendation domains with DCG-like reward structures~\cite{Ie2019SlateQ}.

\textit{2. Position-based model (PBM).}
The classical PBM allows for general formulations of $\mathsf{P}(V|R)$, including the widely adopted functional form $\frac{1}{\log_{2}(i+1)}$.
The recently proposed Contextual PBM~\cite{Fang2019} can be plugged into Eq.~\ref{eq:exposure} to directly provide an unbiased DCG formulation with a context-dependent discount function, enjoying the same theoretical guarantees we have derived for DCG under the PBM.
A variety of other click models~\cite{chuklin2015click} have been proposed in the research literature~\cite{Borisov2016,Chen2020_CACM}, as well as ways to evaluate them~\cite{Deffayet2022}.
We expect that our work provides a basis for further connections to be drawn between click models and unbiased evaluation metrics.

\textit{3. Reward independence across ranks.}
When we do not assume any structure between the actions taken by the ranking policy and the observed rewards, the problem quickly becomes intractable, as we suffer from a combinatorial explosion of the action space.
This is a well-known problem, and the independence assumption has been adopted (either explicitly or implicitly) by a wide array of related work~\cite{Swaminathan2017,Ie2019SlateQ, Bendada2020,Jeunen2021B}.
Note that this assumption does not simply relate to \emph{observing} rewards, but to the underlying distribution of $Q$.
Indeed, related work that adopts a cascading user behaviour model also relies on this assumption, as the cascade relates to the distribution of $V$ (and thus $C$) rather than that of $Q$~\cite{McInerney2020, Kiyohara2022}.
Evaluation metrics have been proposed to encode concepts of listwise novelty and diversity into DCG-like formulations for top-$n$ recommendations~\cite{Clarke2008,Parapar2021}; we conjecture they can be extended to our unbiased setup as well.

\textit{4. The examination hypothesis} implies that \emph{exposure} bias (through $V$) is the main culprit that makes $C$ a noisy indicator of $Q$.
Differences in exposure can then purely come from \emph{position} bias (as in the PBM), but they can also be perpetuated by \emph{selection} bias (as made evident by Eq.~\ref{eq:exposure})~\cite{Diaz2020,Jeunen2021B}.
Other sources of bias have been raised in the literature, such as presentation or trust bias~\cite{Agarwal2019TrustBias,Vardasbi2020}.
We expect that these types of biases can be incorporated into the theoretical derivation of Sec.~\ref{sec:dcg_unbiased} to devise DCG-like formulations that remain unbiased estimators of online reward, even when additional biases are present.
Naturally, such biases are use-case-specific.

\textit{5. Full support of the logging policy.}
The main assumption that makes IPS work, is that no actions with non-zero probability under the target policy can have zero probability under the logging policy.
Indeed, if a context-action pair is known not to be present in the data, we cannot make any inferences about its reward (with guarantees).
This is at the heart of policy-based estimation, but especially problematic in real-world systems where the action space is large and the cost of such \emph{full} randomisation, even with small probabilities, can be high.
Recent work in \emph{learning} from bandit feedback deals with such cases empirically~\cite{JeunenKDD2020,Jeunen2021A} and theoretically~\cite{Sachdeva2020, Lopez2021}, providing a source of inspiration to (partially) alleviate these issues.
\vspace{-4ex}\section{Conclusions \& Outlook}\label{sec:conclusion}
Offline evaluation of recommender systems is a common task, and known to be problematic.
This work investigates the commonly used (normalised) discounted cumulative gain metric and its uses in the research literature.
Specifically, we have investigated \emph{when} we can expect such metrics to approximate the gold standard outcome of an online experiment.
In a counterfactual estimation framework, we formally derived the necessary assumptions to consider DCG an unbiased estimator of online reward.
Whilst it is reassuring that such assumptions exist and we can directly map DCG to online metrics --- we also highlighted how this \emph{ideal} use deviates from the traditional uses of the metric in IR, and how it often appears in the research literature.
We then shifted our focus to \emph{normalised} DCG, and demonstrated its \emph{inconsistency}, both theoretically and empirically with reproducible experiments.
Indeed, even when all neccesary assumptions hold and DCG provides unbiased estimates of online reward, nDCG \emph{cannot} be used to rank competing models, as it does does \emph{not} preserve the rankings we would obtain from DCG.\looseness=-1

Through a correlation analysis between results obtained from off- and on-line experiments on a large-scale recommendation platform, we show that our \emph{unbiased} DCG estimates strongly correlate with online metrics in a real-world use-case. 
Additionally, we show how the offline metric can be used to detect statistically significant online improvements with high sensitivity, further highlighting its promise for offline evaluation in both academia and industry.
Normalised DCG, on the other hand, suffers from a weaker correlation with online results, and lower sensitivity than DCG.
These results suggest that nDCG's practical utility may be limited.

We believe our work opens up interesting areas for future research, where our theoretical framework can be extended to formally assess the assumptions required by other commonly used evaluation metrics in the field.
Furthermore, theoretical and empirical connections between other types of commonly used online evaluation metrics (e.g. user retention) would be fruitful.

\begin{acks}
We are grateful to Lien Michiels, co-author of RecPack~\cite{Michiels2022}, for early feedback and help setting up the experiment in Appendix~\ref{sec:appx}.
\end{acks}

\bibliographystyle{ACM-Reference-Format}
\bibliography{bibliography}


\begin{thebibliography}{105}


\ifx \showCODEN    \undefined \def \showCODEN     #1{\unskip}     \fi
\ifx \showDOI      \undefined \def \showDOI       #1{#1}\fi
\ifx \showISBNx    \undefined \def \showISBNx     #1{\unskip}     \fi
\ifx \showISBNxiii \undefined \def \showISBNxiii  #1{\unskip}     \fi
\ifx \showISSN     \undefined \def \showISSN      #1{\unskip}     \fi
\ifx \showLCCN     \undefined \def \showLCCN      #1{\unskip}     \fi
\ifx \shownote     \undefined \def \shownote      #1{#1}          \fi
\ifx \showarticletitle \undefined \def \showarticletitle #1{#1}   \fi
\ifx \showURL      \undefined \def \showURL       {\relax}        \fi
\providecommand\bibfield[2]{#2}
\providecommand\bibinfo[2]{#2}
\providecommand\natexlab[1]{#1}
\providecommand\showeprint[2][]{arXiv:#2}

\bibitem[Agarwal et~al\mbox{.}(2017)]%
        {Agarwal2017}
\bibfield{author}{\bibinfo{person}{A. Agarwal}, \bibinfo{person}{S. Basu},
  \bibinfo{person}{T. Schnabel}, {and} \bibinfo{person}{T. Joachims}.}
  \bibinfo{year}{2017}\natexlab{}.
\newblock \showarticletitle{Effective Evaluation Using Logged Bandit Feedback
  from Multiple Loggers}. In \bibinfo{booktitle}{\emph{Proc. of the 23rd ACM
  SIGKDD International Conference on Knowledge Discovery \& Data Mining}}
  \emph{(\bibinfo{series}{KDD '17})}. \bibinfo{publisher}{ACM},
  \bibinfo{pages}{687--696}.
\newblock
\showISBNx{978-1-4503-4887-4}
\urldef\tempurl%
\url{https://doi.org/10.1145/3097983.3098155}
\showDOI{\tempurl}


\bibitem[Agarwal et~al\mbox{.}(2019)]%
        {Agarwal2019TrustBias}
\bibfield{author}{\bibinfo{person}{A. Agarwal}, \bibinfo{person}{X. Wang},
  \bibinfo{person}{C. Li}, \bibinfo{person}{M. Bendersky}, {and}
  \bibinfo{person}{M. Najork}.} \bibinfo{year}{2019}\natexlab{}.
\newblock \showarticletitle{Addressing Trust Bias for Unbiased
  Learning-to-Rank}. In \bibinfo{booktitle}{\emph{Proc. of the 2019 World Wide
  Web Conference}} \emph{(\bibinfo{series}{WWW '19})}.
  \bibinfo{publisher}{ACM}, \bibinfo{pages}{4--14}.
\newblock
\showISBNx{978-1-4503-6674-8}
\urldef\tempurl%
\url{https://doi.org/10.1145/3308558.3313697}
\showDOI{\tempurl}


\bibitem[Al-Maskari et~al\mbox{.}(2007)]%
        {AlMaskari2007}
\bibfield{author}{\bibinfo{person}{A. Al-Maskari}, \bibinfo{person}{M.
  Sanderson}, {and} \bibinfo{person}{P. Clough}.}
  \bibinfo{year}{2007}\natexlab{}.
\newblock \showarticletitle{The Relationship between IR Effectiveness Measures
  and User Satisfaction}. In \bibinfo{booktitle}{\emph{Proc of. the 30th Annual
  International ACM SIGIR Conference on Research and Development in Information
  Retrieval}} \emph{(\bibinfo{series}{SIGIR '07})}. \bibinfo{publisher}{ACM},
  \bibinfo{pages}{773–774}.
\newblock
\showISBNx{9781595935977}
\urldef\tempurl%
\url{https://doi.org/10.1145/1277741.1277902}
\showDOI{\tempurl}


\bibitem[Armstrong et~al\mbox{.}(2009)]%
        {Armstrong2009}
\bibfield{author}{\bibinfo{person}{T.~G. Armstrong}, \bibinfo{person}{A.
  Moffat}, \bibinfo{person}{W. Webber}, {and} \bibinfo{person}{J. Zobel}.}
  \bibinfo{year}{2009}\natexlab{}.
\newblock \showarticletitle{Improvements That Don't Add up: Ad-Hoc Retrieval
  Results since 1998}. In \bibinfo{booktitle}{\emph{Proc of. the 18th ACM
  Conference on Information and Knowledge Management}}
  \emph{(\bibinfo{series}{CIKM '09})}. \bibinfo{publisher}{ACM},
  \bibinfo{pages}{601–610}.
\newblock
\showISBNx{9781605585123}
\urldef\tempurl%
\url{https://doi.org/10.1145/1645953.1646031}
\showDOI{\tempurl}


\bibitem[Beel et~al\mbox{.}(2013)]%
        {Beel2013}
\bibfield{author}{\bibinfo{person}{J. Beel}, \bibinfo{person}{M. Genzmehr},
  \bibinfo{person}{S. Langer}, \bibinfo{person}{A. N\"{u}rnberger}, {and}
  \bibinfo{person}{B. Gipp}.} \bibinfo{year}{2013}\natexlab{}.
\newblock \showarticletitle{{A Comparative Analysis of Offline and Online
  Evaluations and Discussion of Research Paper Recommender System Evaluation}}.
  In \bibinfo{booktitle}{\emph{Proc. of the International Workshop on
  Reproducibility and Replication in Recommender Systems Evaluation}}
  \emph{(\bibinfo{series}{RepSys '13})}. \bibinfo{pages}{7--14}.
\newblock
\showISBNx{978-1-4503-2465-6}


\bibitem[Bendada et~al\mbox{.}(2020)]%
        {Bendada2020}
\bibfield{author}{\bibinfo{person}{W. Bendada}, \bibinfo{person}{G. Salha},
  {and} \bibinfo{person}{T. Bontempelli}.} \bibinfo{year}{2020}\natexlab{}.
\newblock \showarticletitle{Carousel Personalization in Music Streaming Apps
  with Contextual Bandits}. In \bibinfo{booktitle}{\emph{Proc. of the 14th ACM
  Conference on Recommender Systems}} \emph{(\bibinfo{series}{RecSys '20})}.
  \bibinfo{publisher}{ACM}, \bibinfo{pages}{420–425}.
\newblock
\urldef\tempurl%
\url{https://doi.org/10.1145/3383313.3412217}
\showDOI{\tempurl}


\bibitem[Bennett et~al\mbox{.}(2007)]%
        {Bennet2007}
\bibfield{author}{\bibinfo{person}{J. Bennett}, \bibinfo{person}{S. Lanning},
  {et~al\mbox{.}}} \bibinfo{year}{2007}\natexlab{}.
\newblock \showarticletitle{The Netflix prize}. In
  \bibinfo{booktitle}{\emph{Proc. of the KDD cup and workshop}},
  Vol.~\bibinfo{volume}{2007}. \bibinfo{pages}{35}.
\newblock


\bibitem[Borisov et~al\mbox{.}(2016)]%
        {Borisov2016}
\bibfield{author}{\bibinfo{person}{A. Borisov}, \bibinfo{person}{I. Markov},
  \bibinfo{person}{M. de Rijke}, {and} \bibinfo{person}{P. Serdyukov}.}
  \bibinfo{year}{2016}\natexlab{}.
\newblock \showarticletitle{A Neural Click Model for Web Search}. In
  \bibinfo{booktitle}{\emph{Proc. of the 25th International Conference on World
  Wide Web}} \emph{(\bibinfo{series}{WWW '16})}. \bibinfo{pages}{531–541}.
\newblock
\showISBNx{9781450341431}
\urldef\tempurl%
\url{https://doi.org/10.1145/2872427.2883033}
\showDOI{\tempurl}


\bibitem[Ca\~{n}amares and Castells(2020)]%
        {Canamares2020RecSys}
\bibfield{author}{\bibinfo{person}{R. Ca\~{n}amares} {and} \bibinfo{person}{P.
  Castells}.} \bibinfo{year}{2020}\natexlab{}.
\newblock \showarticletitle{On Target Item Sampling In Offline Recommender
  System Evaluation}. In \bibinfo{booktitle}{\emph{Proc. of the 14th ACM
  Conference on Recommender Systems}} \emph{(\bibinfo{series}{RecSys '20})}.
  \bibinfo{publisher}{ACM}, \bibinfo{pages}{259–268}.
\newblock
\showISBNx{9781450375832}
\urldef\tempurl%
\url{https://doi.org/10.1145/3383313.3412259}
\showDOI{\tempurl}


\bibitem[Ca{\~{n}}amares et~al\mbox{.}(2020)]%
        {Canamares2020}
\bibfield{author}{\bibinfo{person}{R. Ca{\~{n}}amares}, \bibinfo{person}{P.
  Castells}, {and} \bibinfo{person}{A. Moffat}.}
  \bibinfo{year}{2020}\natexlab{}.
\newblock \showarticletitle{Offline evaluation options for recommender
  systems}.
\newblock \bibinfo{journal}{\emph{Information Retrieval Journal}}
  \bibinfo{volume}{23}, \bibinfo{number}{4} (\bibinfo{date}{01 Aug}
  \bibinfo{year}{2020}), \bibinfo{pages}{387--410}.
\newblock
\showISSN{1573-7659}
\urldef\tempurl%
\url{https://doi.org/10.1007/s10791-020-09371-3}
\showDOI{\tempurl}


\bibitem[Cavenaghi et~al\mbox{.}(2023)]%
        {Cavenaghi2023}
\bibfield{author}{\bibinfo{person}{E. Cavenaghi}, \bibinfo{person}{G.
  Sottocornola}, \bibinfo{person}{F. Stella}, {and} \bibinfo{person}{M.
  Zanker}.} \bibinfo{year}{2023}\natexlab{}.
\newblock \showarticletitle{A Systematic Study on Reproducibility of
  Reinforcement Learning in Recommendation Systems}.
\newblock \bibinfo{journal}{\emph{ACM Trans. Recomm. Syst.}}
  \bibinfo{volume}{1}, \bibinfo{number}{3}, Article \bibinfo{articleno}{11}
  (\bibinfo{date}{jul} \bibinfo{year}{2023}), \bibinfo{numpages}{23}~pages.
\newblock
\urldef\tempurl%
\url{https://doi.org/10.1145/3596519}
\showDOI{\tempurl}


\bibitem[Chapelle et~al\mbox{.}(2009)]%
        {Chapelle2009}
\bibfield{author}{\bibinfo{person}{O. Chapelle}, \bibinfo{person}{D. Metzler},
  \bibinfo{person}{Y. Zhang}, {and} \bibinfo{person}{P. Grinspan}.}
  \bibinfo{year}{2009}\natexlab{}.
\newblock \showarticletitle{Expected Reciprocal Rank for Graded Relevance}. In
  \bibinfo{booktitle}{\emph{Proc of. the 18th ACM Conference on Information and
  Knowledge Management}} \emph{(\bibinfo{series}{CIKM '09})}.
  \bibinfo{publisher}{ACM}, \bibinfo{pages}{621–630}.
\newblock
\showISBNx{9781605585123}
\urldef\tempurl%
\url{https://doi.org/10.1145/1645953.1646033}
\showDOI{\tempurl}


\bibitem[Chen et~al\mbox{.}(2020)]%
        {Chen2020_CACM}
\bibfield{author}{\bibinfo{person}{J. Chen}, \bibinfo{person}{J. Mao},
  \bibinfo{person}{Y. Liu}, \bibinfo{person}{M. Zhang}, {and}
  \bibinfo{person}{S. Ma}.} \bibinfo{year}{2020}\natexlab{}.
\newblock \showarticletitle{A Context-Aware Click Model for Web Search}. In
  \bibinfo{booktitle}{\emph{Proc. of the 13th International Conference on Web
  Search and Data Mining}} \emph{(\bibinfo{series}{WSDM '20})}.
  \bibinfo{publisher}{ACM}, \bibinfo{pages}{88–96}.
\newblock
\showISBNx{9781450368223}
\urldef\tempurl%
\url{https://doi.org/10.1145/3336191.3371819}
\showDOI{\tempurl}


\bibitem[Chen et~al\mbox{.}(2019)]%
        {Chen2019}
\bibfield{author}{\bibinfo{person}{M. Chen}, \bibinfo{person}{A. Beutel},
  \bibinfo{person}{P. Covington}, \bibinfo{person}{S. Jain},
  \bibinfo{person}{F. Belletti}, {and} \bibinfo{person}{E.~H. Chi}.}
  \bibinfo{year}{2019}\natexlab{}.
\newblock \showarticletitle{Top-K Off-Policy Correction for a REINFORCE
  Recommender System}. In \bibinfo{booktitle}{\emph{Proc. of the 12th ACM
  International Conference on Web Search and Data Mining}}
  \emph{(\bibinfo{series}{WSDM '19})}. \bibinfo{publisher}{ACM},
  \bibinfo{pages}{456--464}.
\newblock
\showISBNx{978-1-4503-5940-5}
\urldef\tempurl%
\url{https://doi.org/10.1145/3289600.3290999}
\showDOI{\tempurl}


\bibitem[Chuklin et~al\mbox{.}(2015)]%
        {chuklin2015click}
\bibfield{author}{\bibinfo{person}{A. Chuklin}, \bibinfo{person}{I. Markov},
  {and} \bibinfo{person}{M. {de Rijke}}.} \bibinfo{year}{2015}\natexlab{}.
\newblock \bibinfo{booktitle}{\emph{Click Models for Web Search}}.
\newblock \bibinfo{publisher}{Morgan \& Claypool}.
\newblock
\showISBNx{9781627056489}
\urldef\tempurl%
\url{https://doi.org/10.2200/S00654ED1V01Y201507ICR043}
\showDOI{\tempurl}


\bibitem[Clarke et~al\mbox{.}(2008)]%
        {Clarke2008}
\bibfield{author}{\bibinfo{person}{C.L.A. Clarke}, \bibinfo{person}{M. Kolla},
  \bibinfo{person}{G.~V. Cormack}, \bibinfo{person}{O. Vechtomova},
  \bibinfo{person}{A. Ashkan}, \bibinfo{person}{S. B\"{u}ttcher}, {and}
  \bibinfo{person}{I. MacKinnon}.} \bibinfo{year}{2008}\natexlab{}.
\newblock \showarticletitle{Novelty and Diversity in Information Retrieval
  Evaluation}. In \bibinfo{booktitle}{\emph{Proc. of the 31st Annual
  International ACM SIGIR Conference on Research and Development in Information
  Retrieval}} \emph{(\bibinfo{series}{SIGIR '08})}. \bibinfo{publisher}{ACM},
  \bibinfo{pages}{659–666}.
\newblock
\showISBNx{9781605581644}
\urldef\tempurl%
\url{https://doi.org/10.1145/1390334.1390446}
\showDOI{\tempurl}


\bibitem[Covington et~al\mbox{.}(2016)]%
        {Covington2016}
\bibfield{author}{\bibinfo{person}{P. Covington}, \bibinfo{person}{J. Adams},
  {and} \bibinfo{person}{E. Sargin}.} \bibinfo{year}{2016}\natexlab{}.
\newblock \showarticletitle{Deep Neural Networks for YouTube Recommendations}.
  In \bibinfo{booktitle}{\emph{Proc. of the 10th ACM Conference on Recommender
  Systems}} \emph{(\bibinfo{series}{RecSys '16})}. \bibinfo{publisher}{ACM},
  \bibinfo{pages}{191–198}.
\newblock
\showISBNx{9781450340359}
\urldef\tempurl%
\url{https://doi.org/10.1145/2959100.2959190}
\showDOI{\tempurl}


\bibitem[Craswell et~al\mbox{.}(2008)]%
        {Craswell2008}
\bibfield{author}{\bibinfo{person}{N. Craswell}, \bibinfo{person}{O. Zoeter},
  \bibinfo{person}{M. Taylor}, {and} \bibinfo{person}{B. Ramsey}.}
  \bibinfo{year}{2008}\natexlab{}.
\newblock \showarticletitle{An Experimental Comparison of Click Position-Bias
  Models}. In \bibinfo{booktitle}{\emph{Proc of. the 2008 International
  Conference on Web Search and Data Mining}} \emph{(\bibinfo{series}{WSDM
  '08})}. \bibinfo{publisher}{ACM}, \bibinfo{pages}{87–94}.
\newblock
\showISBNx{9781595939272}
\urldef\tempurl%
\url{https://doi.org/10.1145/1341531.1341545}
\showDOI{\tempurl}


\bibitem[Dang et~al\mbox{.}(2013)]%
        {VanDang2013}
\bibfield{author}{\bibinfo{person}{V. Dang}, \bibinfo{person}{M. Bendersky},
  {and} \bibinfo{person}{W.~B. Croft}.} \bibinfo{year}{2013}\natexlab{}.
\newblock \showarticletitle{Two-Stage Learning to Rank for Information
  Retrieval}. In \bibinfo{booktitle}{\emph{Advances in Information Retrieval}}.
  \bibinfo{publisher}{Springer Berlin Heidelberg}, \bibinfo{pages}{423--434}.
\newblock
\showISBNx{978-3-642-36973-5}


\bibitem[Deffayet et~al\mbox{.}(2022)]%
        {Deffayet2022}
\bibfield{author}{\bibinfo{person}{R. Deffayet}, \bibinfo{person}{J. Renders},
  {and} \bibinfo{person}{M. de Rijke}.} \bibinfo{year}{2022}\natexlab{}.
\newblock \showarticletitle{Evaluating the Robustness of Click Models to Policy
  Distributional Shift}.
\newblock \bibinfo{journal}{\emph{ACM Trans. Inf. Syst.}} (\bibinfo{date}{oct}
  \bibinfo{year}{2022}).
\newblock
\showISSN{1046-8188}
\urldef\tempurl%
\url{https://doi.org/10.1145/3569086}
\showDOI{\tempurl}
\newblock
\shownote{Just Accepted}.


\bibitem[Deffayet et~al\mbox{.}(2023)]%
        {Deffayet2023}
\bibfield{author}{\bibinfo{person}{R. Deffayet}, \bibinfo{person}{T. Thonet},
  \bibinfo{person}{J.~M. Renders}, {and} \bibinfo{person}{M. de Rijke}.}
  \bibinfo{year}{2023}\natexlab{}.
\newblock \showarticletitle{Offline Evaluation for Reinforcement Learning-Based
  Recommendation: A Critical Issue and Some Alternatives}.
\newblock \bibinfo{journal}{\emph{SIGIR Forum}} \bibinfo{volume}{56},
  \bibinfo{number}{2}, Article \bibinfo{articleno}{3} (\bibinfo{date}{jan}
  \bibinfo{year}{2023}), \bibinfo{numpages}{14}~pages.
\newblock
\showISSN{0163-5840}
\urldef\tempurl%
\url{https://doi.org/10.1145/3582900.3582905}
\showDOI{\tempurl}


\bibitem[Diaz(2021)]%
        {Diaz2021}
\bibfield{author}{\bibinfo{person}{F. Diaz}.} \bibinfo{year}{2021}\natexlab{}.
\newblock \showarticletitle{On Evaluating Session-Based Recommendation with
  Implicit Feedback}. In \bibinfo{booktitle}{\emph{Workshop on Perspectives on
  Offline Evaluation for Recommender Systems at RecSys '21 (PERSPECTIVES
  '21)}}.
\newblock


\bibitem[Diaz et~al\mbox{.}(2020)]%
        {Diaz2020}
\bibfield{author}{\bibinfo{person}{F. Diaz}, \bibinfo{person}{B. Mitra},
  \bibinfo{person}{M.~D. Ekstrand}, \bibinfo{person}{A.~J. Biega}, {and}
  \bibinfo{person}{B. Carterette}.} \bibinfo{year}{2020}\natexlab{}.
\newblock \showarticletitle{Evaluating Stochastic Rankings with Expected
  Exposure}. In \bibinfo{booktitle}{\emph{Proc of. the 29th ACM International
  Conference on Information \& Knowledge Management}}
  \emph{(\bibinfo{series}{CIKM '20})}. \bibinfo{publisher}{ACM},
  \bibinfo{pages}{275–284}.
\newblock
\showISBNx{9781450368599}
\urldef\tempurl%
\url{https://doi.org/10.1145/3340531.3411962}
\showDOI{\tempurl}


\bibitem[Dud\'{\i}k et~al\mbox{.}(2011)]%
        {Dudik2011}
\bibfield{author}{\bibinfo{person}{M. Dud\'{\i}k}, \bibinfo{person}{J.
  Langford}, {and} \bibinfo{person}{L. Li}.} \bibinfo{year}{2011}\natexlab{}.
\newblock \showarticletitle{Doubly Robust Policy Evaluation and Learning}. In
  \bibinfo{booktitle}{\emph{Proc. of the 28th International Conference on
  International Conference on Machine Learning}}
  \emph{(\bibinfo{series}{ICML'11})}. \bibinfo{pages}{1097--1104}.
\newblock
\showISBNx{978-1-4503-0619-5}


\bibitem[Elvira et~al\mbox{.}(2019)]%
        {Elvira2019}
\bibfield{author}{\bibinfo{person}{V. Elvira}, \bibinfo{person}{L. Martino},
  \bibinfo{person}{D. Luengo}, {and} \bibinfo{person}{M.F. Bugallo}.}
  \bibinfo{year}{2019}\natexlab{}.
\newblock \showarticletitle{{Generalized Multiple Importance Sampling}}.
\newblock \bibinfo{journal}{\emph{Statist. Sci.}} \bibinfo{volume}{34},
  \bibinfo{number}{1} (\bibinfo{year}{2019}), \bibinfo{pages}{129 -- 155}.
\newblock
\urldef\tempurl%
\url{https://doi.org/10.1214/18-STS668}
\showDOI{\tempurl}


\bibitem[Fang et~al\mbox{.}(2019)]%
        {Fang2019}
\bibfield{author}{\bibinfo{person}{Z. Fang}, \bibinfo{person}{A. Agarwal},
  {and} \bibinfo{person}{T. Joachims}.} \bibinfo{year}{2019}\natexlab{}.
\newblock \showarticletitle{Intervention Harvesting for Context-Dependent
  Examination-Bias Estimation}. In \bibinfo{booktitle}{\emph{Proc of. the 42nd
  International ACM SIGIR Conference on Research and Development in Information
  Retrieval}} \emph{(\bibinfo{series}{SIGIR'19})}. \bibinfo{publisher}{ACM},
  \bibinfo{pages}{825–834}.
\newblock
\showISBNx{9781450361729}
\urldef\tempurl%
\url{https://doi.org/10.1145/3331184.3331238}
\showDOI{\tempurl}


\bibitem[Ferrante et~al\mbox{.}(2021)]%
        {Ferrante2021}
\bibfield{author}{\bibinfo{person}{M. Ferrante}, \bibinfo{person}{N. Ferro},
  {and} \bibinfo{person}{N. Fuhr}.} \bibinfo{year}{2021}\natexlab{}.
\newblock \showarticletitle{Towards Meaningful Statements in IR Evaluation:
  Mapping Evaluation Measures to Interval Scales}.
\newblock \bibinfo{journal}{\emph{IEEE Access}}  \bibinfo{volume}{9}
  (\bibinfo{year}{2021}), \bibinfo{pages}{136182--136216}.
\newblock
\showISSN{2169-3536}
\urldef\tempurl%
\url{https://doi.org/10.1109/ACCESS.2021.3116857}
\showDOI{\tempurl}


\bibitem[Ferrari~Dacrema et~al\mbox{.}(2021)]%
        {FerrariDacrema2021}
\bibfield{author}{\bibinfo{person}{M. Ferrari~Dacrema}, \bibinfo{person}{S.
  Boglio}, \bibinfo{person}{P. Cremonesi}, {and} \bibinfo{person}{D. Jannach}.}
  \bibinfo{year}{2021}\natexlab{}.
\newblock \showarticletitle{A Troubling Analysis of Reproducibility and
  Progress in Recommender Systems Research}.
\newblock \bibinfo{journal}{\emph{ACM Trans. Inf. Syst.}} \bibinfo{volume}{39},
  \bibinfo{number}{2}, Article \bibinfo{articleno}{20} (\bibinfo{date}{jan}
  \bibinfo{year}{2021}), \bibinfo{numpages}{49}~pages.
\newblock
\showISSN{1046-8188}
\urldef\tempurl%
\url{https://doi.org/10.1145/3434185}
\showDOI{\tempurl}


\bibitem[Ferrari~Dacrema et~al\mbox{.}(2019)]%
        {FerrariDacrema2019}
\bibfield{author}{\bibinfo{person}{M. Ferrari~Dacrema}, \bibinfo{person}{P.
  Cremonesi}, {and} \bibinfo{person}{D. Jannach}.}
  \bibinfo{year}{2019}\natexlab{}.
\newblock \showarticletitle{Are We Really Making Much Progress? A Worrying
  Analysis of Recent Neural Recommendation Approaches}. In
  \bibinfo{booktitle}{\emph{Proc. of the 13th ACM Conference on Recommender
  Systems}} \emph{(\bibinfo{series}{RecSys '19})}. \bibinfo{publisher}{ACM},
  \bibinfo{pages}{101–109}.
\newblock
\showISBNx{9781450362436}
\urldef\tempurl%
\url{https://doi.org/10.1145/3298689.3347058}
\showDOI{\tempurl}


\bibitem[Garcin et~al\mbox{.}(2014)]%
        {Garcin2014}
\bibfield{author}{\bibinfo{person}{F. Garcin}, \bibinfo{person}{B. Faltings},
  \bibinfo{person}{O. Donatsch}, \bibinfo{person}{A. Alazzawi},
  \bibinfo{person}{C. Bruttin}, {and} \bibinfo{person}{A. Huber}.}
  \bibinfo{year}{2014}\natexlab{}.
\newblock \showarticletitle{{Offline and Online Evaluation of News Recommender
  Systems at Swissinfo.Ch}}. In \bibinfo{booktitle}{\emph{Proc. of the 8th ACM
  Conference on Recommender Systems}} \emph{(\bibinfo{series}{RecSys '14})}.
  \bibinfo{pages}{169--176}.
\newblock
\showISBNx{978-1-4503-2668-1}
\urldef\tempurl%
\url{https://doi.org/10.1145/2645710.2645745}
\showURL{%
\tempurl}


\bibitem[Gilotte et~al\mbox{.}(2018)]%
        {Gilotte2018}
\bibfield{author}{\bibinfo{person}{A. Gilotte}, \bibinfo{person}{C.
  Calauz\`{e}nes}, \bibinfo{person}{T. Nedelec}, \bibinfo{person}{A. Abraham},
  {and} \bibinfo{person}{S. Doll{\'e}}.} \bibinfo{year}{2018}\natexlab{}.
\newblock \showarticletitle{Offline A/B Testing for Recommender Systems}. In
  \bibinfo{booktitle}{\emph{Proc. of the Eleventh ACM International Conference
  on Web Search and Data Mining}} \emph{(\bibinfo{series}{WSDM '18})}.
  \bibinfo{publisher}{ACM}, \bibinfo{pages}{198--206}.
\newblock
\showISBNx{978-1-4503-5581-0}
\urldef\tempurl%
\url{https://doi.org/10.1145/3159652.3159687}
\showURL{%
\tempurl}


\bibitem[Gruson et~al\mbox{.}(2019)]%
        {Gruson2019}
\bibfield{author}{\bibinfo{person}{A. Gruson}, \bibinfo{person}{P. Chandar},
  \bibinfo{person}{C. Charbuillet}, \bibinfo{person}{J. McInerney},
  \bibinfo{person}{S. Hansen}, \bibinfo{person}{D. Tardieu}, {and}
  \bibinfo{person}{B. Carterette}.} \bibinfo{year}{2019}\natexlab{}.
\newblock \showarticletitle{Offline Evaluation to Make Decisions About Playlist
  Recommendation Algorithms}. In \bibinfo{booktitle}{\emph{Proc of. the 12th
  ACM International Conference on Web Search and Data Mining}}
  \emph{(\bibinfo{series}{WSDM '19})}. \bibinfo{publisher}{ACM},
  \bibinfo{pages}{420--428}.
\newblock
\showISBNx{978-1-4503-5940-5}
\urldef\tempurl%
\url{https://doi.org/10.1145/3289600.3291027}
\showDOI{\tempurl}


\bibitem[Gupta et~al\mbox{.}(2023)]%
        {Gupta2023}
\bibfield{author}{\bibinfo{person}{S. Gupta}, \bibinfo{person}{H. Oosterhuis},
  {and} \bibinfo{person}{M. {de Rijke}}.} \bibinfo{year}{2023}\natexlab{}.
\newblock \showarticletitle{Safe Deployment for Counterfactual Learning to Rank
  with Exposure-Based Risk Minimization}. In \bibinfo{booktitle}{\emph{Proc. of
  the 46th International ACM SIGIR Conference on Research and Development in
  Information Retrieval}} \emph{(\bibinfo{series}{SIGIR '23})}.
  \bibinfo{publisher}{ACM}, \bibinfo{pages}{249–258}.
\newblock
\urldef\tempurl%
\url{https://doi.org/10.1145/3539618.3591760}
\showDOI{\tempurl}


\bibitem[Harper and Konstan(2015)]%
        {Harper2015}
\bibfield{author}{\bibinfo{person}{F.~Maxwell Harper} {and}
  \bibinfo{person}{Joseph~A. Konstan}.} \bibinfo{year}{2015}\natexlab{}.
\newblock \showarticletitle{The MovieLens Datasets: History and Context}.
\newblock \bibinfo{journal}{\emph{ACM Trans. Interact. Intell. Syst.}}
  \bibinfo{volume}{5}, \bibinfo{number}{4}, Article \bibinfo{articleno}{19}
  (\bibinfo{date}{Dec.} \bibinfo{year}{2015}), \bibinfo{numpages}{19}~pages.
\newblock
\showISSN{2160-6455}
\urldef\tempurl%
\url{https://doi.org/10.1145/2827872}
\showDOI{\tempurl}


\bibitem[Herlocker et~al\mbox{.}(2004)]%
        {Herlocker2004}
\bibfield{author}{\bibinfo{person}{J.~L. Herlocker}, \bibinfo{person}{J.~A.
  Konstan}, \bibinfo{person}{L.~G. Terveen}, {and} \bibinfo{person}{J.~T.
  Riedl}.} \bibinfo{year}{2004}\natexlab{}.
\newblock \showarticletitle{Evaluating Collaborative Filtering Recommender
  Systems}.
\newblock \bibinfo{journal}{\emph{ACM Transactions on Information Systems}}
  \bibinfo{volume}{22}, \bibinfo{number}{1} (\bibinfo{date}{Jan.}
  \bibinfo{year}{2004}), \bibinfo{pages}{5--53}.
\newblock
\showISSN{1046-8188}
\urldef\tempurl%
\url{https://doi.org/10.1145/963770.963772}
\showDOI{\tempurl}


\bibitem[Ie et~al\mbox{.}(2019a)]%
        {Ie2019Recsim}
\bibfield{author}{\bibinfo{person}{E. Ie}, \bibinfo{person}{C. Hsu},
  \bibinfo{person}{M. Mladenov}, \bibinfo{person}{V. Jain}, \bibinfo{person}{S.
  Narvekar}, \bibinfo{person}{J. Wang}, \bibinfo{person}{R. Wu}, {and}
  \bibinfo{person}{C. Boutilier}.} \bibinfo{year}{2019}\natexlab{a}.
\newblock \bibinfo{title}{RecSim: A Configurable Simulation Platform for
  Recommender Systems}.
\newblock
\newblock
\urldef\tempurl%
\url{https://arxiv.org/abs/1909.04847}
\showURL{%
\tempurl}


\bibitem[Ie et~al\mbox{.}(2019b)]%
        {Ie2019SlateQ}
\bibfield{author}{\bibinfo{person}{E. Ie}, \bibinfo{person}{V. Jain},
  \bibinfo{person}{J. Wang}, \bibinfo{person}{S. Narvekar}, \bibinfo{person}{R.
  Agarwal}, \bibinfo{person}{R. Wu}, \bibinfo{person}{H. Cheng},
  \bibinfo{person}{T. Chandra}, {and} \bibinfo{person}{C. Boutilier}.}
  \bibinfo{year}{2019}\natexlab{b}.
\newblock \showarticletitle{SlateQ: A tractable decomposition for reinforcement
  learning with recommendation sets}.
\newblock  (\bibinfo{year}{2019}).
\newblock


\bibitem[Ionides(2008)]%
        {Ionides2008}
\bibfield{author}{\bibinfo{person}{E.~L. Ionides}.}
  \bibinfo{year}{2008}\natexlab{}.
\newblock \showarticletitle{Truncated Importance Sampling}.
\newblock \bibinfo{journal}{\emph{Journal of Computational and Graphical
  Statistics}} \bibinfo{volume}{17}, \bibinfo{number}{2}
  (\bibinfo{year}{2008}), \bibinfo{pages}{295--311}.
\newblock


\bibitem[Jadidinejad et~al\mbox{.}(2021)]%
        {Jadidinejad2021}
\bibfield{author}{\bibinfo{person}{A.~H. Jadidinejad}, \bibinfo{person}{C.
  Macdonald}, {and} \bibinfo{person}{I. Ounis}.}
  \bibinfo{year}{2021}\natexlab{}.
\newblock \showarticletitle{The Simpson’s Paradox in the Offline Evaluation
  of Recommendation Systems}.
\newblock \bibinfo{journal}{\emph{ACM Trans. Inf. Syst.}} \bibinfo{volume}{40},
  \bibinfo{number}{1}, Article \bibinfo{articleno}{4} (\bibinfo{date}{sep}
  \bibinfo{year}{2021}), \bibinfo{numpages}{22}~pages.
\newblock
\showISSN{1046-8188}
\urldef\tempurl%
\url{https://doi.org/10.1145/3458509}
\showDOI{\tempurl}


\bibitem[Jagerman et~al\mbox{.}(2022)]%
        {Jagerman2022}
\bibfield{author}{\bibinfo{person}{R. Jagerman}, \bibinfo{person}{X. Wang},
  \bibinfo{person}{H. Zhuang}, \bibinfo{person}{Z. Qin}, \bibinfo{person}{M.
  Bendersky}, {and} \bibinfo{person}{M. Najork}.}
  \bibinfo{year}{2022}\natexlab{}.
\newblock \showarticletitle{Rax: Composable Learning-to-Rank Using JAX}. In
  \bibinfo{booktitle}{\emph{Proc. of the 28th ACM SIGKDD Conference on
  Knowledge Discovery and Data Mining}} \emph{(\bibinfo{series}{KDD '22})}.
  \bibinfo{publisher}{ACM}, \bibinfo{pages}{3051–3060}.
\newblock
\showISBNx{9781450393850}
\urldef\tempurl%
\url{https://doi.org/10.1145/3534678.3539065}
\showDOI{\tempurl}


\bibitem[Jakimov et~al\mbox{.}(2023)]%
        {Jakimov2023}
\bibfield{author}{\bibinfo{person}{M. Jakimov}, \bibinfo{person}{A. Buchholz},
  \bibinfo{person}{Y. Stein}, {and} \bibinfo{person}{T. Joachims}.}
  \bibinfo{year}{2023}\natexlab{}.
\newblock \showarticletitle{Unbiased Offline Evaluation for Learning to Rank
  with Business Rules}. In \bibinfo{booktitle}{\emph{RecSys 2023 Workshop:
  CONSEQUENCES – Causality, Counterfactuals and Sequential Decision-Making}}.
\newblock
\showeprint[arxiv]{2311.01828}


\bibitem[J\"{a}rvelin and Kek\"{a}l\"{a}inen(2002)]%
        {Jarvelin2002}
\bibfield{author}{\bibinfo{person}{K. J\"{a}rvelin} {and} \bibinfo{person}{J.
  Kek\"{a}l\"{a}inen}.} \bibinfo{year}{2002}\natexlab{}.
\newblock \showarticletitle{Cumulated Gain-Based Evaluation of IR Techniques}.
\newblock \bibinfo{journal}{\emph{ACM Trans. Inf. Syst.}} \bibinfo{volume}{20},
  \bibinfo{number}{4} (\bibinfo{date}{oct} \bibinfo{year}{2002}),
  \bibinfo{pages}{422–446}.
\newblock
\showISSN{1046-8188}
\urldef\tempurl%
\url{https://doi.org/10.1145/582415.582418}
\showDOI{\tempurl}


\bibitem[Jeunen(2019)]%
        {Jeunen2019DS}
\bibfield{author}{\bibinfo{person}{O. Jeunen}.}
  \bibinfo{year}{2019}\natexlab{}.
\newblock \showarticletitle{Revisiting Offline Evaluation for Implicit-feedback
  Recommender Systems}. In \bibinfo{booktitle}{\emph{Proc. of the 13th ACM
  Conference on Recommender Systems}} \emph{(\bibinfo{series}{RecSys '19})}.
  \bibinfo{publisher}{ACM}, \bibinfo{pages}{596--600}.
\newblock
\showISBNx{978-1-4503-6243-6}
\urldef\tempurl%
\url{https://doi.org/10.1145/3298689.3347069}
\showURL{%
\tempurl}


\bibitem[Jeunen(2023a)]%
        {Jeunen2023_misassumption}
\bibfield{author}{\bibinfo{person}{O. Jeunen}.}
  \bibinfo{year}{2023}\natexlab{a}.
\newblock \showarticletitle{A Common Misassumption in Online Experiments with
  Machine Learning Models}.
\newblock \bibinfo{journal}{\emph{SIGIR Forum}} \bibinfo{volume}{57},
  \bibinfo{number}{1} (\bibinfo{year}{2023}).
\newblock
\showeprint[arxiv]{2304.10900}~[cs.LG]


\bibitem[Jeunen(2023b)]%
        {Jeunen2023_C3PO}
\bibfield{author}{\bibinfo{person}{O. Jeunen}.}
  \bibinfo{year}{2023}\natexlab{b}.
\newblock \showarticletitle{A Probabilistic Position Bias Model for Short-Video
  Recommendation Feeds}. In \bibinfo{booktitle}{\emph{Proc. of the 17th ACM
  Conference on Recommender Systems}} \emph{(\bibinfo{series}{RecSys '23})}.
  \bibinfo{publisher}{ACM}.
\newblock


\bibitem[Jeunen and Goethals(2021a)]%
        {Jeunen2021A}
\bibfield{author}{\bibinfo{person}{O. Jeunen} {and} \bibinfo{person}{B.
  Goethals}.} \bibinfo{year}{2021}\natexlab{a}.
\newblock \showarticletitle{Pessimistic Reward Models for Off-Policy Learning
  in Recommendation}. In \bibinfo{booktitle}{\emph{Proc. of the Fifteenth ACM
  Conference on Recommender Systems}} \emph{(\bibinfo{series}{RecSys '21})}.
  \bibinfo{publisher}{ACM}, \bibinfo{pages}{63–74}.
\newblock
\urldef\tempurl%
\url{https://doi.org/10.1145/3460231.3474247}
\showURL{%
\tempurl}


\bibitem[Jeunen and Goethals(2021b)]%
        {Jeunen2021B}
\bibfield{author}{\bibinfo{person}{O. Jeunen} {and} \bibinfo{person}{B.
  Goethals}.} \bibinfo{year}{2021}\natexlab{b}.
\newblock \showarticletitle{Top-K Contextual Bandits with Equity of Exposure}.
  In \bibinfo{booktitle}{\emph{Proc of. the 15th ACM Conference on Recommender
  Systems}} \emph{(\bibinfo{series}{RecSys '21})}. \bibinfo{publisher}{ACM},
  \bibinfo{pages}{310–320}.
\newblock
\showISBNx{9781450384582}
\urldef\tempurl%
\url{https://doi.org/10.1145/3460231.3474248}
\showDOI{\tempurl}


\bibitem[Jeunen and Goethals(2023)]%
        {Jeunen2023Pessimism}
\bibfield{author}{\bibinfo{person}{O. Jeunen} {and} \bibinfo{person}{B.
  Goethals}.} \bibinfo{year}{2023}\natexlab{}.
\newblock \showarticletitle{Pessimistic Decision-Making for Recommender
  Systems}.
\newblock \bibinfo{journal}{\emph{ACM Trans. Recomm. Syst.}}
  \bibinfo{volume}{1}, \bibinfo{number}{1}, Article \bibinfo{articleno}{4}
  (\bibinfo{date}{feb} \bibinfo{year}{2023}), \bibinfo{numpages}{27}~pages.
\newblock
\urldef\tempurl%
\url{https://doi.org/10.1145/3568029}
\showDOI{\tempurl}


\bibitem[Jeunen and London(2023)]%
        {Jeunen2023_CONSEQUENCES}
\bibfield{author}{\bibinfo{person}{O. Jeunen} {and} \bibinfo{person}{B.
  London}.} \bibinfo{year}{2023}\natexlab{}.
\newblock \showarticletitle{Offline Recommender System Evaluation under
  Unobserved Confounding}. In \bibinfo{booktitle}{\emph{RecSys 2023 Workshop:
  CONSEQUENCES – Causality, Counterfactuals and Sequential Decision-Making}}.
\newblock
\showeprint[arxiv]{2309.04222}~[cs.LG]


\bibitem[Jeunen et~al\mbox{.}(2019)]%
        {Jeunen2019REVEAL_EVAL}
\bibfield{author}{\bibinfo{person}{O. Jeunen}, \bibinfo{person}{D. Rohde},
  {and} \bibinfo{person}{F. Vasile}.} \bibinfo{year}{2019}\natexlab{}.
\newblock \bibinfo{title}{On the Value of Bandit Feedback for Offline
  Recommender System Evaluation}.
\newblock
\newblock
\showeprint[arxiv]{1907.12384}~[cs.IR]


\bibitem[Jeunen et~al\mbox{.}(2020)]%
        {JeunenKDD2020}
\bibfield{author}{\bibinfo{person}{O. Jeunen}, \bibinfo{person}{D. Rohde},
  \bibinfo{person}{F. Vasile}, {and} \bibinfo{person}{M. Bompaire}.}
  \bibinfo{year}{2020}\natexlab{}.
\newblock \showarticletitle{Joint Policy-Value Learning for Recommendation}. In
  \bibinfo{booktitle}{\emph{Proc. of the 26th ACM SIGKDD International
  Conference on Knowledge Discovery \& Data Mining}}
  \emph{(\bibinfo{series}{KDD '20})}. \bibinfo{publisher}{ACM},
  \bibinfo{pages}{1223--1233}.
\newblock
\urldef\tempurl%
\url{https://doi.org/10.1145/3394486.3403175}
\showURL{%
\tempurl}


\bibitem[Jeunen et~al\mbox{.}(2018)]%
        {Jeunen2018}
\bibfield{author}{\bibinfo{person}{O. Jeunen}, \bibinfo{person}{K. Verstrepen},
  {and} \bibinfo{person}{B. Goethals}.} \bibinfo{year}{2018}\natexlab{}.
\newblock \showarticletitle{Fair Offline Evaluation Methodologies for
  Implicit-feedback Recommender Systems with MNAR Data}. In
  \bibinfo{booktitle}{\emph{Proc. of the REVEAL 18 Workshop on Offline
  Evaluation for Recommender Systems (RecSys '18)}}.
\newblock


\bibitem[Ji et~al\mbox{.}(2023)]%
        {Ji2023}
\bibfield{author}{\bibinfo{person}{Y. Ji}, \bibinfo{person}{A. Sun},
  \bibinfo{person}{J. Zhang}, {and} \bibinfo{person}{C. Li}.}
  \bibinfo{year}{2023}\natexlab{}.
\newblock \showarticletitle{A Critical Study on Data Leakage in Recommender
  System Offline Evaluation}.
\newblock \bibinfo{journal}{\emph{ACM Trans. Inf. Syst.}} \bibinfo{volume}{41},
  \bibinfo{number}{3}, Article \bibinfo{articleno}{75} (\bibinfo{date}{feb}
  \bibinfo{year}{2023}), \bibinfo{numpages}{27}~pages.
\newblock
\showISSN{1046-8188}
\urldef\tempurl%
\url{https://doi.org/10.1145/3569930}
\showDOI{\tempurl}


\bibitem[Joachims et~al\mbox{.}(2021)]%
        {Joachims_London_Su_Swaminathan_Wang_2021}
\bibfield{author}{\bibinfo{person}{T. Joachims}, \bibinfo{person}{B. London},
  \bibinfo{person}{Y. Su}, \bibinfo{person}{A. Swaminathan}, {and}
  \bibinfo{person}{L. Wang}.} \bibinfo{year}{2021}\natexlab{}.
\newblock \showarticletitle{Recommendations as Treatments}.
\newblock \bibinfo{journal}{\emph{AI Magazine}} \bibinfo{volume}{42},
  \bibinfo{number}{3} (\bibinfo{date}{Nov.} \bibinfo{year}{2021}),
  \bibinfo{pages}{19--30}.
\newblock
\urldef\tempurl%
\url{https://doi.org/10.1609/aimag.v42i3.18141}
\showDOI{\tempurl}


\bibitem[Kallus et~al\mbox{.}(2021)]%
        {Kallus2021Optimal}
\bibfield{author}{\bibinfo{person}{N. Kallus}, \bibinfo{person}{Y. Saito},
  {and} \bibinfo{person}{M. Uehara}.} \bibinfo{year}{2021}\natexlab{}.
\newblock \showarticletitle{Optimal Off-Policy Evaluation from Multiple Logging
  Policies}. In \bibinfo{booktitle}{\emph{Proc. of the 38th International
  Conference on Machine Learning}} \emph{(\bibinfo{series}{ICML '21},
  Vol.~\bibinfo{volume}{139})}, \bibfield{editor}{\bibinfo{person}{Marina
  Meila} {and} \bibinfo{person}{Tong Zhang}} (Eds.). \bibinfo{publisher}{PMLR},
  \bibinfo{pages}{5247--5256}.
\newblock
\urldef\tempurl%
\url{https://proceedings.mlr.press/v139/kallus21a.html}
\showURL{%
\tempurl}


\bibitem[Kharitonov et~al\mbox{.}(2017)]%
        {Kharitonov2017}
\bibfield{author}{\bibinfo{person}{E. Kharitonov}, \bibinfo{person}{A. Drutsa},
  {and} \bibinfo{person}{P. Serdyukov}.} \bibinfo{year}{2017}\natexlab{}.
\newblock \showarticletitle{Learning Sensitive Combinations of A/B Test
  Metrics}. In \bibinfo{booktitle}{\emph{Proc. of the Tenth ACM International
  Conference on Web Search and Data Mining}} \emph{(\bibinfo{series}{WSDM
  '17})}. \bibinfo{publisher}{ACM}, \bibinfo{pages}{651–659}.
\newblock
\showISBNx{9781450346757}
\urldef\tempurl%
\url{https://doi.org/10.1145/3018661.3018708}
\showDOI{\tempurl}


\bibitem[Kiyohara et~al\mbox{.}(2022)]%
        {Kiyohara2022}
\bibfield{author}{\bibinfo{person}{H. Kiyohara}, \bibinfo{person}{Y. Saito},
  \bibinfo{person}{T. Matsuhiro}, \bibinfo{person}{Y. Narita},
  \bibinfo{person}{N. Shimizu}, {and} \bibinfo{person}{Y. Yamamoto}.}
  \bibinfo{year}{2022}\natexlab{}.
\newblock \showarticletitle{Doubly Robust Off-Policy Evaluation for Ranking
  Policies under the Cascade Behavior Model}. In \bibinfo{booktitle}{\emph{Proc
  of. the Fifteenth ACM International Conference on Web Search and Data
  Mining}} \emph{(\bibinfo{series}{WSDM '22})}. \bibinfo{publisher}{ACM},
  \bibinfo{pages}{487–497}.
\newblock
\urldef\tempurl%
\url{https://doi.org/10.1145/3488560.3498380}
\showDOI{\tempurl}


\bibitem[Kohavi et~al\mbox{.}(2022)]%
        {Kohavi2022}
\bibfield{author}{\bibinfo{person}{R. Kohavi}, \bibinfo{person}{A. Deng}, {and}
  \bibinfo{person}{L. Vermeer}.} \bibinfo{year}{2022}\natexlab{}.
\newblock \showarticletitle{A/B Testing Intuition Busters: Common
  Misunderstandings in Online Controlled Experiments}. In
  \bibinfo{booktitle}{\emph{Proc. of the 28th ACM SIGKDD Conference on
  Knowledge Discovery and Data Mining}} \emph{(\bibinfo{series}{KDD '22})}.
  \bibinfo{publisher}{ACM}, \bibinfo{pages}{3168–3177}.
\newblock
\showISBNx{9781450393850}
\urldef\tempurl%
\url{https://doi.org/10.1145/3534678.3539160}
\showDOI{\tempurl}


\bibitem[Kohavi et~al\mbox{.}(2020)]%
        {Kohavi2020}
\bibfield{author}{\bibinfo{person}{R. Kohavi}, \bibinfo{person}{D. Tang}, {and}
  \bibinfo{person}{Y. Xu}.} \bibinfo{year}{2020}\natexlab{}.
\newblock \bibinfo{booktitle}{\emph{Trustworthy online controlled experiments:
  A practical guide to A/B testing}}.
\newblock \bibinfo{publisher}{Cambridge University Press}.
\newblock


\bibitem[Krichene and Rendle(2020)]%
        {Krichene2020}
\bibfield{author}{\bibinfo{person}{W. Krichene} {and} \bibinfo{person}{S.
  Rendle}.} \bibinfo{year}{2020}\natexlab{}.
\newblock \showarticletitle{On Sampled Metrics for Item Recommendation}. In
  \bibinfo{booktitle}{\emph{Proc. of the 26th ACM SIGKDD International
  Conference on Knowledge Discovery \& Data Mining}}
  \emph{(\bibinfo{series}{KDD '20})}. \bibinfo{publisher}{ACM},
  \bibinfo{pages}{1748–1757}.
\newblock
\showISBNx{9781450379984}
\urldef\tempurl%
\url{https://doi.org/10.1145/3394486.3403226}
\showDOI{\tempurl}


\bibitem[Larsen et~al\mbox{.}(2023)]%
        {Larsen2023}
\bibfield{author}{\bibinfo{person}{N. Larsen}, \bibinfo{person}{J. Stallrich},
  \bibinfo{person}{S. Sengupta}, \bibinfo{person}{A. Deng}, \bibinfo{person}{R.
  Kohavi}, {and} \bibinfo{person}{N.~T. Stevens}.}
  \bibinfo{year}{2023}\natexlab{}.
\newblock \showarticletitle{Statistical Challenges in Online Controlled
  Experiments: A Review of A/B Testing Methodology}.
\newblock \bibinfo{journal}{\emph{The American Statistician}}
  \bibinfo{volume}{0}, \bibinfo{number}{0} (\bibinfo{year}{2023}),
  \bibinfo{pages}{1--15}.
\newblock
\urldef\tempurl%
\url{https://doi.org/10.1080/00031305.2023.2257237}
\showDOI{\tempurl}


\bibitem[Li et~al\mbox{.}(2020)]%
        {Li2020}
\bibfield{author}{\bibinfo{person}{D. Li}, \bibinfo{person}{R. Jin},
  \bibinfo{person}{J. Gao}, {and} \bibinfo{person}{Z. Liu}.}
  \bibinfo{year}{2020}\natexlab{}.
\newblock \showarticletitle{On Sampling Top-K Recommendation Evaluation}. In
  \bibinfo{booktitle}{\emph{Proc of. the 26th ACM SIGKDD International
  Conference on Knowledge Discovery \& Data Mining}}
  \emph{(\bibinfo{series}{KDD '20})}. \bibinfo{publisher}{ACM},
  \bibinfo{pages}{2114–2124}.
\newblock
\showISBNx{9781450379984}
\urldef\tempurl%
\url{https://doi.org/10.1145/3394486.3403262}
\showDOI{\tempurl}


\bibitem[Lopez et~al\mbox{.}(2021)]%
        {Lopez2021}
\bibfield{author}{\bibinfo{person}{R. Lopez}, \bibinfo{person}{I.~S. Dhillon},
  {and} \bibinfo{person}{M.~I. Jordan}.} \bibinfo{year}{2021}\natexlab{}.
\newblock \showarticletitle{Learning from eXtreme Bandit Feedback}.
\newblock \bibinfo{journal}{\emph{Proc of. the AAAI Conference on Artificial
  Intelligence}} \bibinfo{volume}{35}, \bibinfo{number}{10}
  (\bibinfo{date}{May} \bibinfo{year}{2021}), \bibinfo{pages}{8732--8740}.
\newblock
\urldef\tempurl%
\url{https://doi.org/10.1609/aaai.v35i10.17058}
\showDOI{\tempurl}


\bibitem[Lyzhin et~al\mbox{.}(2023)]%
        {Lyzhin2023}
\bibfield{author}{\bibinfo{person}{I. Lyzhin}, \bibinfo{person}{A. Ustimenko},
  \bibinfo{person}{A. Gulin}, {and} \bibinfo{person}{L. Prokhorenkova}.}
  \bibinfo{year}{2023}\natexlab{}.
\newblock \showarticletitle{Which Tricks are Important for Learning to Rank?}.
  In \bibinfo{booktitle}{\emph{Proc of. the 40th International Conference on
  Machine Learning}} \emph{(\bibinfo{series}{ICML '23'},
  Vol.~\bibinfo{volume}{202})}. \bibinfo{publisher}{PMLR},
  \bibinfo{pages}{23264--23278}.
\newblock
\urldef\tempurl%
\url{https://proceedings.mlr.press/v202/lyzhin23a.html}
\showURL{%
\tempurl}


\bibitem[Ma et~al\mbox{.}(2020)]%
        {Ma2020}
\bibfield{author}{\bibinfo{person}{J. Ma}, \bibinfo{person}{Z. Zhao},
  \bibinfo{person}{X. Yi}, \bibinfo{person}{J. Yang}, \bibinfo{person}{M.
  Chen}, \bibinfo{person}{J. Tang}, \bibinfo{person}{L. Hong}, {and}
  \bibinfo{person}{E.~H. Chi}.} \bibinfo{year}{2020}\natexlab{}.
\newblock \showarticletitle{Off-Policy Learning in Two-Stage Recommender
  Systems}. In \bibinfo{booktitle}{\emph{Proc. of the 2020 World Wide Web
  Conference}} \emph{(\bibinfo{series}{WWW '20})}. \bibinfo{publisher}{ACM}.
\newblock
\urldef\tempurl%
\url{https://doi.org/10.1145/3366423.3380130}
\showDOI{\tempurl}


\bibitem[McInerney et~al\mbox{.}(2020)]%
        {McInerney2020}
\bibfield{author}{\bibinfo{person}{J. McInerney}, \bibinfo{person}{B. Brost},
  \bibinfo{person}{P. Chandar}, \bibinfo{person}{R. Mehrotra}, {and}
  \bibinfo{person}{B. Carterette}.} \bibinfo{year}{2020}\natexlab{}.
\newblock \showarticletitle{Counterfactual Evaluation of Slate Recommendations
  with Sequential Reward Interactions}. In \bibinfo{booktitle}{\emph{Proc of.
  the 26th ACM SIGKDD International Conference on Knowledge Discovery \& Data
  Mining}} \emph{(\bibinfo{series}{KDD '20})}. \bibinfo{publisher}{ACM},
  \bibinfo{pages}{1779–1788}.
\newblock
\showISBNx{9781450379984}
\urldef\tempurl%
\url{https://doi.org/10.1145/3394486.3403229}
\showDOI{\tempurl}


\bibitem[Mehrotra et~al\mbox{.}(2018)]%
        {Mehrotra2018}
\bibfield{author}{\bibinfo{person}{R. Mehrotra}, \bibinfo{person}{J.
  McInerney}, \bibinfo{person}{H. Bouchard}, \bibinfo{person}{M. Lalmas}, {and}
  \bibinfo{person}{F. Diaz}.} \bibinfo{year}{2018}\natexlab{}.
\newblock \showarticletitle{Towards a Fair Marketplace: Counterfactual
  Evaluation of the Trade-off between Relevance, Fairness \& Satisfaction in
  Recommendation Systems}. In \bibinfo{booktitle}{\emph{Proc. of the 27th ACM
  International Conference on Information and Knowledge Management}}
  \emph{(\bibinfo{series}{CIKM '18})}. \bibinfo{publisher}{ACM},
  \bibinfo{pages}{2243–2251}.
\newblock
\showISBNx{9781450360142}
\urldef\tempurl%
\url{https://doi.org/10.1145/3269206.3272027}
\showDOI{\tempurl}


\bibitem[Mehrotra et~al\mbox{.}(2020)]%
        {Mehrotra2020}
\bibfield{author}{\bibinfo{person}{R. Mehrotra}, \bibinfo{person}{N. Xue},
  {and} \bibinfo{person}{M. Lalmas}.} \bibinfo{year}{2020}\natexlab{}.
\newblock \showarticletitle{Bandit Based Optimization of Multiple Objectives on
  a Music Streaming Platform}. In \bibinfo{booktitle}{\emph{Proc of. the 26th
  ACM SIGKDD International Conference on Knowledge Discovery \& Data Mining}}
  \emph{(\bibinfo{series}{KDD '20})}. \bibinfo{publisher}{ACM},
  \bibinfo{pages}{3224–3233}.
\newblock
\showISBNx{9781450379984}
\urldef\tempurl%
\url{https://doi.org/10.1145/3394486.3403374}
\showDOI{\tempurl}


\bibitem[Michiels et~al\mbox{.}(2022)]%
        {Michiels2022}
\bibfield{author}{\bibinfo{person}{L. Michiels}, \bibinfo{person}{R.
  Verachtert}, {and} \bibinfo{person}{B. Goethals}.}
  \bibinfo{year}{2022}\natexlab{}.
\newblock \showarticletitle{RecPack: An(Other) Experimentation Toolkit for
  Top-N Recommendation Using Implicit Feedback Data}. In
  \bibinfo{booktitle}{\emph{Proc. of the 16th ACM Conference on Recommender
  Systems}} \emph{(\bibinfo{series}{RecSys '22})}. \bibinfo{publisher}{ACM},
  \bibinfo{pages}{648–651}.
\newblock
\showISBNx{9781450392785}
\urldef\tempurl%
\url{https://doi.org/10.1145/3523227.3551472}
\showDOI{\tempurl}


\bibitem[Moffat and Zobel(2008)]%
        {Moffat2008}
\bibfield{author}{\bibinfo{person}{A. Moffat} {and} \bibinfo{person}{J.
  Zobel}.} \bibinfo{year}{2008}\natexlab{}.
\newblock \showarticletitle{Rank-Biased Precision for Measurement of Retrieval
  Effectiveness}.
\newblock \bibinfo{journal}{\emph{ACM Trans. Inf. Syst.}} \bibinfo{volume}{27},
  \bibinfo{number}{1}, Article \bibinfo{articleno}{2} (\bibinfo{date}{dec}
  \bibinfo{year}{2008}), \bibinfo{numpages}{27}~pages.
\newblock
\showISSN{1046-8188}
\urldef\tempurl%
\url{https://doi.org/10.1145/1416950.1416952}
\showDOI{\tempurl}


\bibitem[Oosterhuis(2022)]%
        {Oosterhuis2022}
\bibfield{author}{\bibinfo{person}{H. Oosterhuis}.}
  \bibinfo{year}{2022}\natexlab{}.
\newblock \showarticletitle{Learning-to-Rank at the Speed of Sampling:
  Plackett-Luce Gradient Estimation with Minimal Computational Complexity}. In
  \bibinfo{booktitle}{\emph{Proc of. the 45th International ACM SIGIR
  Conference on Research and Development in Information Retrieval}}
  \emph{(\bibinfo{series}{SIGIR '22})}. \bibinfo{publisher}{ACM},
  \bibinfo{pages}{2266–2271}.
\newblock
\urldef\tempurl%
\url{https://doi.org/10.1145/3477495.3531842}
\showDOI{\tempurl}


\bibitem[Oosterhuis(2023)]%
        {Oosterhuis2023}
\bibfield{author}{\bibinfo{person}{H. Oosterhuis}.}
  \bibinfo{year}{2023}\natexlab{}.
\newblock \showarticletitle{Doubly Robust Estimation for Correcting Position
  Bias in Click Feedback for Unbiased Learning to Rank}.
\newblock \bibinfo{journal}{\emph{ACM Trans. Inf. Syst.}} \bibinfo{volume}{41},
  \bibinfo{number}{3}, Article \bibinfo{articleno}{61} (\bibinfo{date}{feb}
  \bibinfo{year}{2023}), \bibinfo{numpages}{33}~pages.
\newblock
\showISSN{1046-8188}
\urldef\tempurl%
\url{https://doi.org/10.1145/3569453}
\showDOI{\tempurl}


\bibitem[Oosterhuis and de~Rijke(2020)]%
        {Oosterhuis2020}
\bibfield{author}{\bibinfo{person}{H. Oosterhuis} {and} \bibinfo{person}{M. de
  Rijke}.} \bibinfo{year}{2020}\natexlab{}.
\newblock \showarticletitle{Policy-Aware Unbiased Learning to Rank for Top-k
  Rankings}. In \bibinfo{booktitle}{\emph{Proc of. the 43rd International ACM
  SIGIR Conference on Research and Development in Information Retrieval}}
  \emph{(\bibinfo{series}{SIGIR '20})}. \bibinfo{publisher}{ACM},
  \bibinfo{pages}{489–498}.
\newblock
\showISBNx{9781450380164}
\urldef\tempurl%
\url{https://doi.org/10.1145/3397271.3401102}
\showDOI{\tempurl}


\bibitem[Owen(2013)]%
        {Owen2013}
\bibfield{author}{\bibinfo{person}{A.~B. Owen}.}
  \bibinfo{year}{2013}\natexlab{}.
\newblock \bibinfo{booktitle}{\emph{Monte Carlo theory, methods and examples}}.
\newblock


\bibitem[Parapar and Radlinski(2021)]%
        {Parapar2021}
\bibfield{author}{\bibinfo{person}{J. Parapar} {and} \bibinfo{person}{F.
  Radlinski}.} \bibinfo{year}{2021}\natexlab{}.
\newblock \showarticletitle{Towards Unified Metrics for Accuracy and Diversity
  for Recommender Systems}. In \bibinfo{booktitle}{\emph{Proc. of the 15th ACM
  Conference on Recommender Systems}} \emph{(\bibinfo{series}{RecSys '21})}.
  \bibinfo{publisher}{ACM}, \bibinfo{pages}{75–84}.
\newblock
\showISBNx{9781450384582}
\urldef\tempurl%
\url{https://doi.org/10.1145/3460231.3474234}
\showDOI{\tempurl}


\bibitem[Rendle et~al\mbox{.}(2020)]%
        {Rendle2020}
\bibfield{author}{\bibinfo{person}{S. Rendle}, \bibinfo{person}{W. Krichene},
  \bibinfo{person}{L. Zhang}, {and} \bibinfo{person}{J. Anderson}.}
  \bibinfo{year}{2020}\natexlab{}.
\newblock \showarticletitle{Neural Collaborative Filtering vs. Matrix
  Factorization Revisited}. In \bibinfo{booktitle}{\emph{Proc of. the 14th ACM
  Conference on Recommender Systems}} \emph{(\bibinfo{series}{RecSys '20})}.
  \bibinfo{publisher}{ACM}, \bibinfo{pages}{240–248}.
\newblock
\showISBNx{9781450375832}
\urldef\tempurl%
\url{https://doi.org/10.1145/3383313.3412488}
\showDOI{\tempurl}


\bibitem[Rendle et~al\mbox{.}(2022)]%
        {Rendle2022}
\bibfield{author}{\bibinfo{person}{S. Rendle}, \bibinfo{person}{W. Krichene},
  \bibinfo{person}{L. Zhang}, {and} \bibinfo{person}{Y. Koren}.}
  \bibinfo{year}{2022}\natexlab{}.
\newblock \showarticletitle{Revisiting the Performance of IALS on Item
  Recommendation Benchmarks}. In \bibinfo{booktitle}{\emph{Proc. of the 16th
  ACM Conference on Recommender Systems}} \emph{(\bibinfo{series}{RecSys
  '22})}. \bibinfo{publisher}{ACM}, \bibinfo{pages}{427–435}.
\newblock
\showISBNx{9781450392785}
\urldef\tempurl%
\url{https://doi.org/10.1145/3523227.3548486}
\showDOI{\tempurl}


\bibitem[Rendle et~al\mbox{.}(2019)]%
        {Rendle2019}
\bibfield{author}{\bibinfo{person}{S. Rendle}, \bibinfo{person}{L. Zhang},
  {and} \bibinfo{person}{Y. Koren}.} \bibinfo{year}{2019}\natexlab{}.
\newblock \bibinfo{title}{{On the Difficulty of Evaluating Baselines: A Study
  on Recommender Systems}}.
\newblock
\newblock
\showeprint[arxiv]{1905.01395}~[cs.IR]


\bibitem[Rohde et~al\mbox{.}(2018)]%
        {rohde2018recogym}
\bibfield{author}{\bibinfo{person}{D. Rohde}, \bibinfo{person}{S. Bonner},
  \bibinfo{person}{T. Dunlop}, \bibinfo{person}{F. Vasile}, {and}
  \bibinfo{person}{A. Karatzoglou}.} \bibinfo{year}{2018}\natexlab{}.
\newblock \showarticletitle{RecoGym: A Reinforcement Learning Environment for
  the problem of Product Recommendation in Online Advertising}.
\newblock \bibinfo{journal}{\emph{arXiv preprint arXiv:1808.00720}}
  (\bibinfo{year}{2018}).
\newblock


\bibitem[Rossetti et~al\mbox{.}(2016)]%
        {Rossetti2016}
\bibfield{author}{\bibinfo{person}{M. Rossetti}, \bibinfo{person}{F. Stella},
  {and} \bibinfo{person}{M. Zanker}.} \bibinfo{year}{2016}\natexlab{}.
\newblock \showarticletitle{{Contrasting Offline and Online Results when
  Evaluating Recommendation Algorithms}}. In \bibinfo{booktitle}{\emph{Proc. of
  the 10th ACM Conference on Recommender Systems}}
  \emph{(\bibinfo{series}{RecSys '16})}. \bibinfo{publisher}{ACM},
  \bibinfo{pages}{31--34}.
\newblock
\showISBNx{978-1-4503-4035-9}
\urldef\tempurl%
\url{https://doi.org/10.1145/2959100.2959176}
\showURL{%
\tempurl}


\bibitem[Sachdeva et~al\mbox{.}(2020)]%
        {Sachdeva2020}
\bibfield{author}{\bibinfo{person}{N. Sachdeva}, \bibinfo{person}{Y. Su}, {and}
  \bibinfo{person}{T. Joachims}.} \bibinfo{year}{2020}\natexlab{}.
\newblock \showarticletitle{Off-Policy Bandits with Deficient Support}. In
  \bibinfo{booktitle}{\emph{Proc of. the 26th ACM SIGKDD International
  Conference on Knowledge Discovery \& Data Mining}}
  \emph{(\bibinfo{series}{KDD '20})}. \bibinfo{publisher}{ACM},
  \bibinfo{pages}{965–975}.
\newblock
\showISBNx{9781450379984}
\urldef\tempurl%
\url{https://doi.org/10.1145/3394486.3403139}
\showDOI{\tempurl}


\bibitem[Saito et~al\mbox{.}(2021a)]%
        {Saito2021OBP}
\bibfield{author}{\bibinfo{person}{Y. Saito}, \bibinfo{person}{S. Aihara},
  \bibinfo{person}{M. Matsutani}, {and} \bibinfo{person}{Y. Narita}.}
  \bibinfo{year}{2021}\natexlab{a}.
\newblock \showarticletitle{Open Bandit Dataset and Pipeline: Towards Realistic
  and Reproducible Off-Policy Evaluation}. In \bibinfo{booktitle}{\emph{Proc
  of. the Neural Information Processing Systems Track on Datasets and
  Benchmarks}}, Vol.~\bibinfo{volume}{1}.
\newblock
\urldef\tempurl%
\url{https://datasets-benchmarks-proceedings.neurips.cc/paper_files/paper/2021/file/33e75ff09dd601bbe69f351039152189-Paper-round2.pdf}
\showURL{%
\tempurl}


\bibitem[Saito and Joachims(2021)]%
        {Saito2021}
\bibfield{author}{\bibinfo{person}{Y. Saito} {and} \bibinfo{person}{T.
  Joachims}.} \bibinfo{year}{2021}\natexlab{}.
\newblock \showarticletitle{Counterfactual Learning and Evaluation for
  Recommender Systems: Foundations, Implementations, and Recent Advances}. In
  \bibinfo{booktitle}{\emph{Proc. of the 15th ACM Conference on Recommender
  Systems}} \emph{(\bibinfo{series}{RecSys '21})}. \bibinfo{publisher}{ACM},
  \bibinfo{pages}{828–830}.
\newblock
\showISBNx{9781450384582}
\urldef\tempurl%
\url{https://doi.org/10.1145/3460231.3473320}
\showDOI{\tempurl}


\bibitem[Saito and Joachims(2022)]%
        {Saito2022_ICML}
\bibfield{author}{\bibinfo{person}{Y. Saito} {and} \bibinfo{person}{T.
  Joachims}.} \bibinfo{year}{2022}\natexlab{}.
\newblock \showarticletitle{Off-Policy Evaluation for Large Action Spaces via
  Embeddings}. In \bibinfo{booktitle}{\emph{Proc. of the 39th International
  Conference on Machine Learning}} \emph{(\bibinfo{series}{ICML '22},
  Vol.~\bibinfo{volume}{162})}. \bibinfo{publisher}{PMLR},
  \bibinfo{pages}{19089--19122}.
\newblock
\urldef\tempurl%
\url{https://proceedings.mlr.press/v162/saito22a.html}
\showURL{%
\tempurl}


\bibitem[Saito et~al\mbox{.}(2023)]%
        {Saito2023_ICML}
\bibfield{author}{\bibinfo{person}{Y. Saito}, \bibinfo{person}{Q. Ren}, {and}
  \bibinfo{person}{T. Joachims}.} \bibinfo{year}{2023}\natexlab{}.
\newblock \showarticletitle{Off-Policy Evaluation for Large Action Spaces via
  Conjunct Effect Modeling}. In \bibinfo{booktitle}{\emph{Proc. of the 40th
  International Conference on Machine Learning}} \emph{(\bibinfo{series}{ICML
  '23}, Vol.~\bibinfo{volume}{202})}. \bibinfo{publisher}{PMLR},
  \bibinfo{pages}{29734--29759}.
\newblock
\urldef\tempurl%
\url{https://proceedings.mlr.press/v202/saito23b.html}
\showURL{%
\tempurl}


\bibitem[Saito et~al\mbox{.}(2021b)]%
        {Saito2021_Robustness}
\bibfield{author}{\bibinfo{person}{Y. Saito}, \bibinfo{person}{T. Udagawa},
  \bibinfo{person}{H. Kiyohara}, \bibinfo{person}{K. Mogi}, \bibinfo{person}{Y.
  Narita}, {and} \bibinfo{person}{K. Tateno}.}
  \bibinfo{year}{2021}\natexlab{b}.
\newblock \showarticletitle{Evaluating the Robustness of Off-Policy
  Evaluation}. In \bibinfo{booktitle}{\emph{Proc. of the 15th ACM Conference on
  Recommender Systems}} \emph{(\bibinfo{series}{RecSys '21})}.
  \bibinfo{publisher}{ACM}, \bibinfo{pages}{114–123}.
\newblock
\showISBNx{9781450384582}
\urldef\tempurl%
\url{https://doi.org/10.1145/3460231.3474245}
\showDOI{\tempurl}


\bibitem[Sakhi et~al\mbox{.}(2020)]%
        {Sakhi2020}
\bibfield{author}{\bibinfo{person}{O. Sakhi}, \bibinfo{person}{S. Bonner},
  \bibinfo{person}{D. Rohde}, {and} \bibinfo{person}{F. Vasile}.}
  \bibinfo{year}{2020}\natexlab{}.
\newblock \showarticletitle{BLOB : A Probabilistic Model for Recommendation
  that Combines Organic and Bandit Signals}. In \bibinfo{booktitle}{\emph{Proc.
  of the 26th ACM Conference on Knowledge Discovery \& Data Mining}}
  \emph{(\bibinfo{series}{KDD '20})}. \bibinfo{publisher}{ACM}.
\newblock
\urldef\tempurl%
\url{https://doi.org/10.1145/3394486.3403121}
\showDOI{\tempurl}


\bibitem[Sch{\"u}tze et~al\mbox{.}(2008)]%
        {schutze2008introduction}
\bibfield{author}{\bibinfo{person}{H. Sch{\"u}tze}, \bibinfo{person}{C.~D
  Manning}, {and} \bibinfo{person}{P. Raghavan}.}
  \bibinfo{year}{2008}\natexlab{}.
\newblock \bibinfo{booktitle}{\emph{Introduction to information retrieval}}.
  Vol.~\bibinfo{volume}{39}.
\newblock \bibinfo{publisher}{Cambridge University Press Cambridge}.
\newblock


\bibitem[Steck(2013)]%
        {Steck2013}
\bibfield{author}{\bibinfo{person}{H. Steck}.} \bibinfo{year}{2013}\natexlab{}.
\newblock \showarticletitle{Evaluation of Recommendations: Rating-prediction
  and Ranking}. In \bibinfo{booktitle}{\emph{Proc. of the 7th ACM Conference on
  Recommender Systems}} \emph{(\bibinfo{series}{RecSys '13})}.
  \bibinfo{publisher}{ACM}, \bibinfo{pages}{213--220}.
\newblock
\showISBNx{978-1-4503-2409-0}
\urldef\tempurl%
\url{https://doi.org/10.1145/2507157.2507160}
\showDOI{\tempurl}


\bibitem[Steck(2019)]%
        {Steck2019}
\bibfield{author}{\bibinfo{person}{H. Steck}.} \bibinfo{year}{2019}\natexlab{}.
\newblock \showarticletitle{Embarrassingly Shallow Autoencoders for Sparse
  Data}. In \bibinfo{booktitle}{\emph{The World Wide Web Conference}}
  \emph{(\bibinfo{series}{WWW ’19})}. \bibinfo{publisher}{ACM},
  \bibinfo{pages}{3251–3257}.
\newblock
\showISBNx{9781450366748}


\bibitem[Student(1908)]%
        {Student1908}
\bibfield{author}{\bibinfo{person}{Student}.} \bibinfo{year}{1908}\natexlab{}.
\newblock \showarticletitle{Probable Error of a Correlation Coefficient}.
\newblock \bibinfo{journal}{\emph{Biometrika}} \bibinfo{volume}{6},
  \bibinfo{number}{2/3} (\bibinfo{year}{1908}), \bibinfo{pages}{302--310}.
\newblock
\showISSN{00063444}
\urldef\tempurl%
\url{http://www.jstor.org/stable/2331474}
\showURL{%
\tempurl}


\bibitem[Su et~al\mbox{.}(2020a)]%
        {Su2020_ICML}
\bibfield{author}{\bibinfo{person}{Y. Su}, \bibinfo{person}{M. Dimakopoulou},
  \bibinfo{person}{A. Krishnamurthy}, {and} \bibinfo{person}{M. Dudik}.}
  \bibinfo{year}{2020}\natexlab{a}.
\newblock \showarticletitle{Doubly robust off-policy evaluation with
  shrinkage}. In \bibinfo{booktitle}{\emph{Proc. of the 37th International
  Conference on Machine Learning}} \emph{(\bibinfo{series}{ICML '20})}.
  \bibinfo{publisher}{PMLR}, \bibinfo{pages}{9167--9176}.
\newblock


\bibitem[Su et~al\mbox{.}(2020b)]%
        {Su2020_Adaptive}
\bibfield{author}{\bibinfo{person}{Y. Su}, \bibinfo{person}{P. Srinath}, {and}
  \bibinfo{person}{A. Krishnamurthy}.} \bibinfo{year}{2020}\natexlab{b}.
\newblock \showarticletitle{Adaptive Estimator Selection for Off-Policy
  Evaluation}. In \bibinfo{booktitle}{\emph{Proc of. the 37th International
  Conference on Machine Learning}} \emph{(\bibinfo{series}{ICML '20},
  Vol.~\bibinfo{volume}{119})}. \bibinfo{publisher}{PMLR},
  \bibinfo{pages}{9196--9205}.
\newblock
\urldef\tempurl%
\url{https://proceedings.mlr.press/v119/su20d.html}
\showURL{%
\tempurl}


\bibitem[Su et~al\mbox{.}(2019)]%
        {su2019cab}
\bibfield{author}{\bibinfo{person}{Y. Su}, \bibinfo{person}{L. Wang},
  \bibinfo{person}{M. Santacatterina}, {and} \bibinfo{person}{T. Joachims}.}
  \bibinfo{year}{2019}\natexlab{}.
\newblock \showarticletitle{CAB: Continuous Adaptive Blending for Policy
  Evaluation and Learning}. In \bibinfo{booktitle}{\emph{International
  Conference on Machine Learning}} \emph{(\bibinfo{series}{ICML'19})}.
  \bibinfo{pages}{6005--6014}.
\newblock


\bibitem[Swaminathan and Joachims(2015)]%
        {Swaminathan2015snips}
\bibfield{author}{\bibinfo{person}{A. Swaminathan} {and} \bibinfo{person}{T.
  Joachims}.} \bibinfo{year}{2015}\natexlab{}.
\newblock \showarticletitle{The Self-Normalized Estimator for Counterfactual
  Learning}. In \bibinfo{booktitle}{\emph{Advances in Neural Information
  Processing Systems}}. \bibinfo{pages}{3231--3239}.
\newblock
\urldef\tempurl%
\url{https://proceedings.neurips.cc/paper/2015/file/39027dfad5138c9ca0c474d71db915c3-Paper.pdf}
\showURL{%
\tempurl}


\bibitem[Swaminathan et~al\mbox{.}(2017)]%
        {Swaminathan2017}
\bibfield{author}{\bibinfo{person}{A. Swaminathan}, \bibinfo{person}{A.
  Krishnamurthy}, \bibinfo{person}{A. Agarwal}, \bibinfo{person}{M. Dudik},
  \bibinfo{person}{J. Langford}, \bibinfo{person}{D. Jose}, {and}
  \bibinfo{person}{I. Zitouni}.} \bibinfo{year}{2017}\natexlab{}.
\newblock \showarticletitle{Off-policy evaluation for slate recommendation}. In
  \bibinfo{booktitle}{\emph{Advances in Neural Information Processing
  Systems}}, Vol.~\bibinfo{volume}{30}. \bibinfo{publisher}{Curran Associates,
  Inc.}
\newblock
\urldef\tempurl%
\url{https://proceedings.neurips.cc/paper/2017/file/5352696a9ca3397beb79f116f3a33991-Paper.pdf}
\showURL{%
\tempurl}


\bibitem[Tucker and Joachims(2023)]%
        {Tucker2023}
\bibfield{author}{\bibinfo{person}{A.~D. Tucker} {and} \bibinfo{person}{T.
  Joachims}.} \bibinfo{year}{2023}\natexlab{}.
\newblock \showarticletitle{Variance-Minimizing Augmentation Logging for
  Counterfactual Evaluation in Contextual Bandits}. In
  \bibinfo{booktitle}{\emph{Proc of. the Sixteenth ACM International Conference
  on Web Search and Data Mining}} \emph{(\bibinfo{series}{WSDM '23})}.
  \bibinfo{publisher}{ACM}, \bibinfo{pages}{967–975}.
\newblock
\showISBNx{9781450394079}
\urldef\tempurl%
\url{https://doi.org/10.1145/3539597.3570452}
\showDOI{\tempurl}


\bibitem[Udagawa et~al\mbox{.}(2023)]%
        {Udagawa2022}
\bibfield{author}{\bibinfo{person}{T. Udagawa}, \bibinfo{person}{H. Kiyohara},
  \bibinfo{person}{Y. Narita}, \bibinfo{person}{Y. Saito}, {and}
  \bibinfo{person}{K. Tateno}.} \bibinfo{year}{2023}\natexlab{}.
\newblock \showarticletitle{Policy-Adaptive Estimator Selection for Off-Policy
  Evaluation}.
\newblock \bibinfo{journal}{\emph{Proc. of the AAAI Conference on Artificial
  Intelligence}} \bibinfo{volume}{37}, \bibinfo{number}{8}
  (\bibinfo{date}{Jun.} \bibinfo{year}{2023}), \bibinfo{pages}{10025--10033}.
\newblock
\urldef\tempurl%
\url{https://doi.org/10.1609/aaai.v37i8.26195}
\showDOI{\tempurl}


\bibitem[Ustimenko and Prokhorenkova(2020)]%
        {Ustimenko2020}
\bibfield{author}{\bibinfo{person}{A. Ustimenko} {and} \bibinfo{person}{L.
  Prokhorenkova}.} \bibinfo{year}{2020}\natexlab{}.
\newblock \showarticletitle{{S}tochastic{R}ank: Global Optimization of
  Scale-Free Discrete Functions}. In \bibinfo{booktitle}{\emph{Proc of. the
  37th International Conference on Machine Learning}}
  \emph{(\bibinfo{series}{ICML '20'}, Vol.~\bibinfo{volume}{119})}.
  \bibinfo{publisher}{PMLR}, \bibinfo{pages}{9669--9679}.
\newblock
\urldef\tempurl%
\url{https://proceedings.mlr.press/v119/ustimenko20a.html}
\showURL{%
\tempurl}


\bibitem[Valcarce et~al\mbox{.}(2020)]%
        {Valcarce2020}
\bibfield{author}{\bibinfo{person}{D. Valcarce}, \bibinfo{person}{A.
  Bellog{\'i}n}, \bibinfo{person}{J. Parapar}, {and} \bibinfo{person}{P.
  Castells}.} \bibinfo{year}{2020}\natexlab{}.
\newblock \showarticletitle{Assessing ranking metrics in top-N recommendation}.
\newblock \bibinfo{journal}{\emph{Information Retrieval Journal}}
  \bibinfo{volume}{23}, \bibinfo{number}{4} (\bibinfo{date}{01 Aug}
  \bibinfo{year}{2020}), \bibinfo{pages}{411--448}.
\newblock
\showISSN{1573-7659}
\urldef\tempurl%
\url{https://doi.org/10.1007/s10791-020-09377-x}
\showDOI{\tempurl}


\bibitem[Vardasbi et~al\mbox{.}(2020)]%
        {Vardasbi2020}
\bibfield{author}{\bibinfo{person}{A. Vardasbi}, \bibinfo{person}{H.
  Oosterhuis}, {and} \bibinfo{person}{M. de Rijke}.}
  \bibinfo{year}{2020}\natexlab{}.
\newblock \showarticletitle{When Inverse Propensity Scoring Does Not Work:
  Affine Corrections for Unbiased Learning to Rank}. In
  \bibinfo{booktitle}{\emph{Proc of. the 29th ACM International Conference on
  Information \& Knowledge Management}} \emph{(\bibinfo{series}{CIKM '20})}.
  \bibinfo{publisher}{ACM}, \bibinfo{pages}{1475–1484}.
\newblock
\showISBNx{9781450368599}
\urldef\tempurl%
\url{https://doi.org/10.1145/3340531.3412031}
\showDOI{\tempurl}


\bibitem[Vasile et~al\mbox{.}(2020)]%
        {Vasile2020}
\bibfield{author}{\bibinfo{person}{F. Vasile}, \bibinfo{person}{D. Rohde},
  \bibinfo{person}{O. Jeunen}, {and} \bibinfo{person}{A. Benhalloum}.}
  \bibinfo{year}{2020}\natexlab{}.
\newblock \showarticletitle{A Gentle Introduction to Recommendation as
  Counterfactual Policy Learning}. In \bibinfo{booktitle}{\emph{Proc. of the
  28th ACM Conference on User Modeling, Adaptation and Personalization}}
  \emph{(\bibinfo{series}{UMAP '20})}. \bibinfo{publisher}{ACM},
  \bibinfo{pages}{392–393}.
\newblock
\showISBNx{9781450368612}
\urldef\tempurl%
\url{https://doi.org/10.1145/3340631.3398666}
\showDOI{\tempurl}


\bibitem[Verstrepen and Goethals(2014)]%
        {Verstrepen2014}
\bibfield{author}{\bibinfo{person}{K. Verstrepen} {and} \bibinfo{person}{B.
  Goethals}.} \bibinfo{year}{2014}\natexlab{}.
\newblock \showarticletitle{Unifying Nearest Neighbors Collaborative
  Filtering}. In \bibinfo{booktitle}{\emph{Proc. of the 8th ACM Conference on
  Recommender Systems}} \emph{(\bibinfo{series}{RecSys '14})}.
  \bibinfo{publisher}{ACM}, \bibinfo{pages}{177–184}.
\newblock
\showISBNx{9781450326681}
\urldef\tempurl%
\url{https://doi.org/10.1145/2645710.2645731}
\showDOI{\tempurl}


\bibitem[Yang et~al\mbox{.}(2018)]%
        {LYang2018}
\bibfield{author}{\bibinfo{person}{L. Yang}, \bibinfo{person}{Y. Cui},
  \bibinfo{person}{Yuan X.}, \bibinfo{person}{C. Wang}, \bibinfo{person}{S.
  Belongie}, {and} \bibinfo{person}{D. Estrin}.}
  \bibinfo{year}{2018}\natexlab{}.
\newblock \showarticletitle{Unbiased Offline Recommender Evaluation for
  Missing-not-at-random Implicit Feedback}. In \bibinfo{booktitle}{\emph{Proc.
  of the 12th ACM Conference on Recommender Systems}}
  \emph{(\bibinfo{series}{RecSys '18})}. \bibinfo{publisher}{ACM},
  \bibinfo{pages}{279--287}.
\newblock
\showISBNx{978-1-4503-5901-6}
\urldef\tempurl%
\url{https://doi.org/10.1145/3240323.3240355}
\showDOI{\tempurl}


\bibitem[Zangerle and Bauer(2022)]%
        {Zangerle2023}
\bibfield{author}{\bibinfo{person}{E. Zangerle} {and} \bibinfo{person}{C.
  Bauer}.} \bibinfo{year}{2022}\natexlab{}.
\newblock \showarticletitle{Evaluating Recommender Systems: Survey and
  Framework}.
\newblock \bibinfo{journal}{\emph{ACM Comput. Surv.}} \bibinfo{volume}{55},
  \bibinfo{number}{8}, Article \bibinfo{articleno}{170} (\bibinfo{date}{dec}
  \bibinfo{year}{2022}), \bibinfo{numpages}{38}~pages.
\newblock
\showISSN{0360-0300}
\urldef\tempurl%
\url{https://doi.org/10.1145/3556536}
\showDOI{\tempurl}


\end{thebibliography}

\appendix
\begin{figure}[!t]
    \centering
    \vspace{-2ex}
    \includegraphics[width=\linewidth]{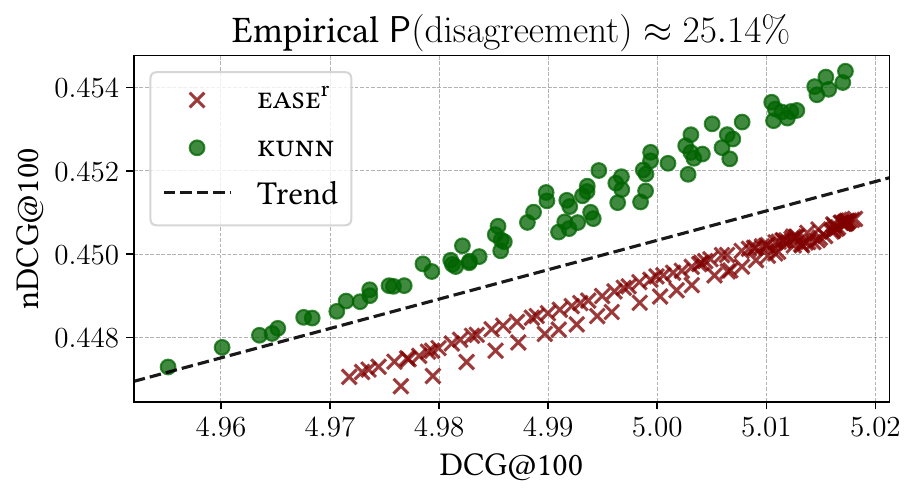}
    \caption{DCG and nDCG exhibit significant disagreement for a standard offline evaluation setup on MovieLens-1M.}
    \label{fig:disagreement}
\end{figure}

\section{Empirical Evidence of (n)DCG Inconsistency on Public Data}\label{sec:appx}
Table~\ref{tab:proof} provides a formal proof that an ordering over competing recommendation (or IR) models obtained through a normalised metric is \emph{not} guaranteed to be consistent with the original metric.
Nevertheless, one might wonder whether this single example represents a misguided pathological case, or whether metric disagreement occurs in practice.
This gives rise to the research question:
\begin{description}    
    \item[\textbf{RQ5}] \textit{Do DCG and normalised DCG disagree when ranking recommendation models in typical offline evaluation setups?}
\end{description}
To answer this question, we make use of the RecPack Python package~\cite{Michiels2022} and the MovieLens-1M dataset~\cite{Harper2015}.
We consider two types of models, \textsc{ease}\textsuperscript{r}~\cite{Steck2019} and \textsc{kunn}~\cite{Verstrepen2014}, varying their hyperparameters to train 192 models on a fixed 50\% of the available user-item interactions, and assess their performance on the held-out 50\%.
This style of evaluation setup is prevalent in the recommendation field~\cite{Steck2013,Jeunen2019DS,Zangerle2023}.
We adopt this package, dataset and methods to provide a reproducible setup that runs in under 20 minutes on a 2021 MacBook Pro.
All source code, including hyperparameter ranges, is available at \href{https://github.com/olivierjeunen/nDCG-disagreement/}{github.com/olivierjeunen/nDCG-disagreement}.

We do not de-bias the interactions (as MovieLens does not provide information about exposure), and adopt the traditional logarithmic discount for DCG with a cut-off at rank 100.
Results are visualised in Figure~\ref{fig:disagreement} with DCG@100 on the $x$-axis and nDCG@100 on the $y$-axis.
The two metrics exhibit a linear correlation of $\approx 0.6$ (Pearson), and a rank correlation of $\approx 0.5$ (Kendall).
Whilst they are clearly correlated, practitioners should \emph{not} blindly adopt nDCG when DCG estimates their online metric.
Indeed, DCG can be formulated as an unbiased estimator of the average reward per trajectory, but nDCG cannot.
As can be seen from the plot, significant \emph{disagreement} occurs between the two metrics: when randomly choosing two observations, the empirical probability of nDCG inverting the ordering implied by DCG is roughly 25\% on this example.
Naturally, one would expect this type of disagreement to occur even more frequently when considering Learning-to-Rank algorithms that directly optimise listwise objectives such as (n)DCG~\cite{Jagerman2022,Ustimenko2020,Lyzhin2023}.

Note that this discrepancy would not occur if we would sample the exact same number of held-out items for every user (as in Leave-One-Out Cross-Validation).
Indeed, in such cases $f^{\star}(x)$ is constant $\forall x \in \mathcal{X}$, simply rescaling the metric.
Whilst this practice can be common in academic scenarios, real-world use-cases typically imply varying numbers of ``relevant'' items per user or context.

We include these results to aid in the reproducibility of the empirical phenomena we report in this work.

\end{document}